\documentclass{article}
\usepackage{arxiv}

\usepackage{times}
\usepackage{soul}
\usepackage{url}
\usepackage[hidelinks]{hyperref}
\usepackage[utf8]{inputenc}
\usepackage[small]{caption}
\usepackage{graphicx}
\usepackage{amsmath}
\usepackage{amsthm}
\usepackage{amssymb}
\usepackage{booktabs}
\usepackage{algorithm}
\usepackage{algpseudocode}
\usepackage{proof}
\usepackage{tikz}
\usetikzlibrary{positioning}
\usepackage{dblfloatfix}
\usepackage{lipsum}

\usepackage{todonotes}

\newtheorem{example}{Example}
\newtheorem{theorem}{Theorem}

\newcommand\ldiaarg[1]{\langle#1\rangle}

\newcommand{\LL}{\mathcal{L}}
\newcommand{\M}{\mathcal{M}}

\newcommand{\cT}{\mathcal{T}}

\newcommand{\Ag}{Agt}
\newcommand{\BP}{\mathcal{P}}
\newcommand{\prefixes}{\mathit{Pre}}

\renewcommand{\phi}{\varphi}

\newcommand{\Exp}{\it Exp} 

\newcommand{\imp}{\rightarrow}

\newcommand{\POL}{\mathsf{POL}}
\newcommand{\EL}{\mathsf{EL}}
\newcommand{\DEL}{\mathsf{DEL}}

\newcommand{\POLS}{\mathsf{POL}^-}
\newcommand{\tP}{\mathsf{P}}

\newcommand{\EPDL}{\mbox{\rm EPDL}}
\newcommand{\PAL}{\mathsf{PAL}}

\newcommand{\union}{\cup}

\newcommand{\obsright}{\blacktriangleright}
\newcommand{\obsleft}{\blacktriangleleft}
\newcommand{\obsup}{\blacktriangle}
\newcommand{\obsdown}{\blacktriangledown}

\newcommand{\expwatere}{(\obsright + \obsup)^{\leq 3} (\obsdown + \obsleft + \ep) (\obsright + \obsup)^{\leq 3}}
\newcommand{\exppower}{(\obsleft + \obsdown)^{\leq 3}}
\newcommand{\expwater}{(\obsright + \obsup)^{\leq 3}}
\newcommand{\Act}{\ensuremath{\mathbf{\Sigma}}}
\newcommand{\vbot}{\mathit{vbot}}

\newcommand{\tr}{\mathsf{tr}}
\newcommand{\starfree}{\mathsf{Star\mbox{-}Free}}
\newcommand{\word}{\mathsf{Word}}

\newcommand{\PSPACE}{\mathsf{PSPACE}}

\newcommand{\NEXPTIME}{\mathsf{NEXPTIME}}

\newcommand{\NP}{\mathsf{NP}}

\newcommand{\modelM}{\mathcal M}

\newcommand{\set}[1]{\{#1\}}
\newcommand{\suchthat}{\mid}

\newcommand{\ep}{\ensuremath{\varepsilon}}
\newcommand{\regdiv}[1]{\ensuremath{\backslash} #1}

\newtheorem{proposition}[theorem]{Proposition}
\newtheorem{corollary}[theorem]{Corollary}

\newtheorem{claim}[theorem]{Claim}
\newtheorem{defi}[theorem]{Definition}

\newcommand{\tile}[6]{
	\draw[fill=#3] (#1,#2) -- (#1,#2+1) -- (#1+0.5, #2+0.5) -- (#1, #2);
	\draw[fill=#4] (#1,#2+1)-- (#1+1,#2+1) -- (#1 +0.5, #2+0.5) -- (#1, #2+1);
	\draw[fill=#5] (#1+1,#2)-- (#1+1,#2+1) -- (#1 +0.5, #2+0.5) -- (#1 +1, #2);
	\draw[fill=#6] (#1,#2)-- (#1 +1,#2) -- (#1 +0.5, #2+0.5) -- (#1, #2);
	\draw[fill=black] (#1+0.5, #2+0.5) -- (#1+0.6, #2+0.4) -- (#1+0.4, #2+0.4) -- cycle;
}

\newcommand{\tilewcsmall}[4]{\tikz[scale=0.5]{\tile{0}{0}{#1}{#2}{#3}{#4}}}

\newcommand{\tilered}{red!80}

\newcommand{\tilegreen}{green!70}

\newcommand{\tilewhite}{white}

\newcommand{\tileseed}{t_0}
\newcommand{\tileset}{T}

\title{On simple expectations and observations of intelligent agents: A complexity study}

\author{
Sourav Chakraborty \\
Indian Statistical Institute, Kolkata, India\\
sourav@isical.ac.in\\
\And
Avijeet Ghosh\\
Indian Statistical Institute, Kolkata, India\\
avijeet\_r@isical.ac.in\\
\And
Sujata Ghosh \\
Indian Statistical Institute, Chennai, India\\
sujata@isichennai.res.in\\
\And
Fran{\c{c}}ois Schwarzentruber\\
Unive Rennes, IRISA, France\\
francois.schwarzentruber@ens-rennes.fr\\
}

\makeatletter
\newenvironment{breakablealgorithm}
  {
   \begin{center}
     \refstepcounter{algorithm}
     \hrule height.8pt depth0pt \kern2pt
     \renewcommand{\caption}[2][\relax]{
       {\raggedright\textbf{\fname@algorithm~\thealgorithm} ##2\par}%
       \ifx\relax##1\relax 
         \addcontentsline{loa}{algorithm}{\protect\numberline{\thealgorithm}##2}%
       \else 
         \addcontentsline{loa}{algorithm}{\protect\numberline{\thealgorithm}##1}%
       \fi
       \kern2pt\hrule\kern2pt
     }
  }{
     \kern2pt\hrule\relax
   \end{center}
  }
\makeatother

\algnewcommand\algorithmicforeach{\textbf{for each}}
\algdef{S}[FOR]{ForEach}[1]{\algorithmicforeach\ #1\ \algorithmicdo}

\begin{document}

\maketitle

\begin{abstract}
  Public observation logic (POL) reasons about agent expectations and agent observations in various real world situations. The expectations of agents take shape based on certain protocols about the world around and they remove those possible scenarios where their expectations and observations do not match. This in turn influences the epistemic reasoning of these agents. In this work, we study the computational complexity of the satisfaction problems of various fragments of POL. In the process, we also highlight the inevitable link that these fragments have with the well-studied Public announcement logic.
\end{abstract}

\section{Introduction}

Reasoning about knowledge among multiple agents plays an important role in studying real-world problems in a distributed setting, e.g., in communicating processes, protocols, strategies and games. \textit{Multi-agent epistemic logic ($\EL$)} \cite{RAK} and its dynamic extensions, popularly known as \textit{dynamic epistemic logics ($\DEL$)} \cite{DEL} are well-known logical systems to specify and reason about such dynamic interactions of knowledge. Traditionally, agents' knowledge is about facts and $\EL$/$\DEL$ mostly deals with this phenomenon of `knowing that'. More recently, the notions of `knowing whether', `knowing why' and `knowing how' have also been investigated from a formal viewpoint \cite{wang2018}. 

These agents also have expectations about the world around them, and they reason based on what they observe around them, and such observations may or may not match the expectations they have about their surroundings. Following \cite{DBLP:conf/icla/Wang11}, such perspectives on agent reasoning were taken up by \cite{van2014hidden} and studied formally in the form of \textit{Public observation logic ($\POL$)}. We present below a situation that $\POL$ is adept at modelling. The example is in the lines of the one considered in \cite{DBLP:conf/ijcai/0001GGS22}:

\begin{example}\label{ex-intro}
Let us consider a robotic vacuum cleaner ($\vbot$) that is moving on a floor represented as a $7\times 7$ grid (see Figure~\ref{figure:field}). On the top right of the floor, there is a debris-disposal area, and on the bottom left, there is a power source to recharge. Two children Alice and Bob are awed by this new robotic cleaner. They are watching it move and trying to guess which direction it is moving. The system is adaptive, thus the global behaviour is not hard-coded but learned. We suppose that $\vbot$ moves on a grid and the children may observe one of the four directions: right ($\obsright$), left ($\obsleft$), up($\obsup$) or down($\obsdown$), and of course, combinations of them. Note that, for example, observing $\obsleft$ means that the bot moves one step left. Let Alice be aware of a glitch in the bot. Then her expectations regarding the $\vbot$'s movements include the following possibilities:
 \begin{enumerate}
     \item The bot may go up or right for debris-disposal, but may make an erroneous move, that is, a down or a left move.
     \item The bot may go towards power source without error.
 \end{enumerate}
The only difference between Bob's expectation and that of Alice is that Bob does not consider the bot to make an error while moving towards debris-disposal since he is unaware of the glitch.

Suppose the $\vbot$ is indeed moving towards power from the center of the grid. Hence if the bot makes one left move, $\obsleft$, Bob would know that the bot is moving towards power whereas Alice would still consider moving towards debris-disposal a possibility.

\begin{figure}
	\begin{center}
		\newcommand{\sizefield}{7}
		\begin{tikzpicture}[scale=0.3]
			\foreach \x in {0, 1, ..., \sizefield} {
				\draw (\x, 0) -- (\x, \sizefield);
				\draw (0, \x) -- (\sizefield, \x);
			}
			\node at (0.5, 0.5) {\includegraphics[width=0.3cm]{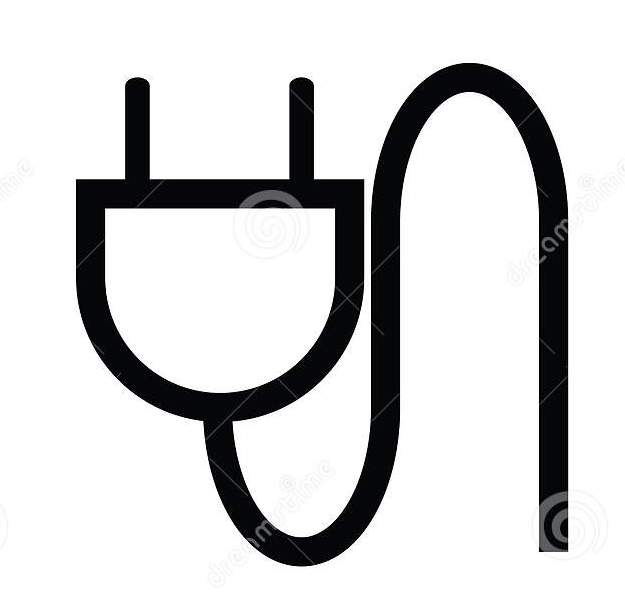}};
			\node at (6.5, 6.5) {\includegraphics[width=0.3cm]{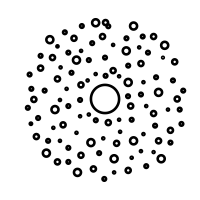}};
			\node at (3.5, 3.5) {\includegraphics[width=0.3cm]{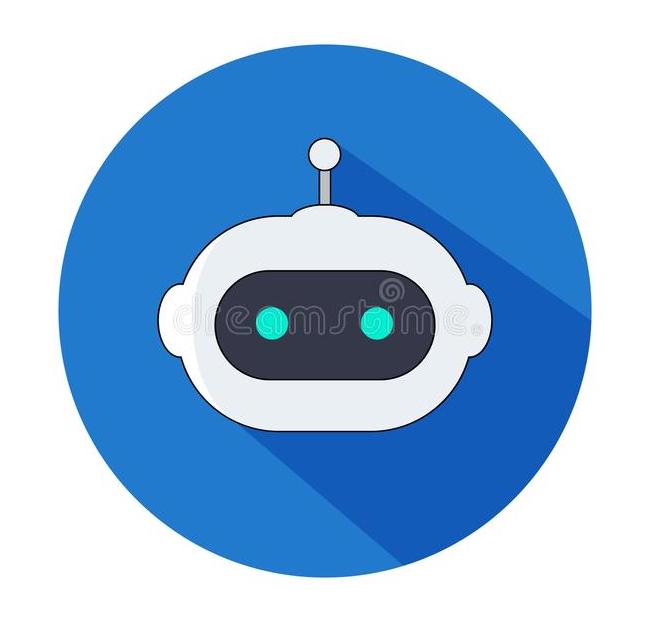}};
		\end{tikzpicture}
	\end{center}
	\vspace{-3mm}
\caption{A robotic vacuum cleaner on the floor (in the middle of the grid). The power source is at bottom left, whereas the debris-disposal area is at top right.\label{figure:field}}
\end{figure}
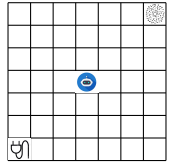

\end{example}

 The example concerns certain rules that we follow in our daily life, they deal with situations where agents expect certain observations at certain states based on some pre-defined \textit{protocols}, viz. the bot mechanism in the example given above. They get to know about the actual situation by observing certain actions which agree with their expectations corresponding to that situation. $\POL$ does not deal with the protocols themselves, but the effect those protocols have in our understanding of the world around us in terms of our expectations and observations. In \cite{DBLP:conf/ijcai/0001GGS22} we have investigated the computational complexity of the model-checking problem of different fragments of $\POL$, and in this paper, we will deal with the computational complexity of the satisfaction problem of various proper fragments of $\POL$ (cf. Figure \ref{figure:results}). We will show how certain simple fragments of $\POL$ give rise to high complexity with respect to their computational behaviour. 
 
To prove the complexity results of some fragment(s) of $\POL$ we use a translation to Public announcement logic ($\PAL$) \cite{DBLP:journals/synthese/Plaza07}, whereas, for other fragment(s), a tableau method is utilized where the tableau rules provide a mix of modal logic reasoning and computations of language theory residuals.

\emph{Outline. }
In Section \ref{section:background}, we recall the relevant definitions of $\POL$. In Section \ref{section:application}, we describe an application of the satisfiability problem of $\POLS$. In Section \ref{section:starfreetableau} we present a $\NEXPTIME$ algorithm for $\POLS$ using the tableau method. In Section \ref{section:nexptimehard}, we prove that $\POLS$ is in $\NEXPTIME$-Hard. In section \ref{section:fragmentcomplexity}, we present the complexity results for various fragments of $\POLS$. 
Section \ref{section:relatedwork} discusses related work, and Section \ref{section:perspectives} concludes the paper.

\begin{figure}[t]
    \hspace*{0pt}
 	\begin{center}
 \begin{tabular}{l|l|l}
  & Single-agent & Multi-agent \\
  \hline
Word $\POL^-$ & $\NP$-Complete & $\PSPACE$-Complete\\
\hline
$\POL^-$ & $\PSPACE$-Hard & $\NEXPTIME$-Complete
 \end{tabular}\
 \end{center}
	\caption{Complexity results of satisfiability of various fragments of $\POL^-$.\label{figure:results}}
\end{figure}
 


 
 


\

\section{Background}\label{section:background}
In this section, we provide a brief overview of a fragment of public observation logic ($\POL$) \cite{van2014hidden}, which we term as $\POL^-$.   


\subsection{A fragment of \texorpdfstring{$\POL (\POL^-)$}{POLmIn}}

Let $\Ag$ be a finite set of agents,  $\BP$ be a countable set of propositions describing the facts about the state and $\Act$ be a finite set of actions. 

An \textit{observation} is a finite string of actions. In the vacuum bot example, an observation may be $\obsleft\obsdown\obsright\obsup$ and similar others. An agent may expect different potential observations to happen at a given state, but to model human/agent expectations, such expectations are described in a finitary way by introducing the {\em observation expressions} (as star-free regular expressions over $\Act$): 

\begin{defi}[Observation expressions]
Given a finite set of action symbols $\Act$, the language $\mathcal{L}_{\it obs}$ of {\em observation expressions} is defined by the following BNF:
$$\begin{array}{r@{\quad::= \quad}l}
  \pi  &
           \emptyset\mid\
           \ep\mid
	  a
 	      \mid \pi\cdot \pi
          \mid \pi + \pi\\
\end{array}$$ 
\noindent where $\emptyset$ denotes the empty set of observations,  the constant~$\ep$ represents the empty string, and $a\in\Act$.
\label{definition:obsexpression}
\end{defi}

In the bot example, the observation expression $(\obsleft\cdot\obsdown + \obsright\cdot\obsup)$ models the expectation of the bot's movement in either way, towards the power source or the debris-disposal area, whereas $(\obsleft)^3\cdot(\obsdown)^3$ models the expectation of moving towards the power source. 

The size of an observation expression $\pi$ is denoted by $|\pi|$. The semantics for the observation expressions are given by {\em sets of observations} (strings over $\Act$), similar to those for regular expressions. Given an observation expression $\pi$, 
its
{\em set of observations} is denoted by $\LL(\pi)$. For example, $\LL(\obsright)=\{\obsright\}$, and $\LL(\obsleft\cdot\obsdown + \obsright\cdot\obsup) = \{\obsleft\obsdown, \obsright\obsup\}$. The (star-free) regular language $\pi\regdiv w$ is the set of words given by  $\{v \in \Act^* \mid wv\in\LL(\pi)\}$. The language $\prefixes(\pi)$ is the set of prefixes of words in $\LL(\pi)$, that is, $w\in \prefixes(\pi)$ 
			iff $\exists v\in \Act^*$ such that $wv\in\LL(\pi)$ (namely, $\LL(\pi\regdiv w)\not=\emptyset$).
%

\begin{example}
  $(\obsleft\cdot\obsdown)\regdiv\obsleft = (\obsleft\cdot\obsdown + \obsright\cdot\obsup)\regdiv\obsleft = \obsdown$, and $\prefixes(\obsleft\cdot\obsdown + \obsright\cdot\obsup) = \{\ep, \obsleft, \obsleft\obsdown, \obsright, \obsright\obsup\}.$
\end{example}

We now present a modified version of {\em epistemic expectation models} from \cite{van2014hidden} that capture the expected observations of agents. They can be seen as epistemic models together with, for each state, a set of potential or expected observations. Recall that an epistemic model is a tuple $\langle S, \sim ,V\rangle$ where $S$ is a non-empty set of states, $\sim$ assigns to each agent in $\Ag$ an equivalence relation $\sim_i \subseteq S \times S$, and $V : S \rightarrow 2^{\BP}$ is a valuation function.

\begin{defi}[Epistemic expectation model with finite observations]
\label{Model1}
An {\em epistemic expectation model with finite observations} $\M$ is a quadruple $\langle S, \sim ,V, \Exp\rangle,$ where $\langle S, \sim ,V\rangle$ is an epistemic model (the {\em epistemic skeleton} of $\M$) and  $\Exp : S \imp \LL_{\it obs}$ is an expected observation function assigning to each state an observation expression $\pi$ such that $\LL(\pi)\not=\emptyset$ (finite non-empty set of finite sequences of observations).   
A pointed epistemic expectation model with finite observations is a pair $(\M, s)$ where $\M = \langle S, \sim ,V, \Exp\rangle$ is an epistemic expectation model with finite observations and $s \in S$. In what follows we will use the `epistemic expectation model' to denote the `epistemic expectation model with finite observations'.
\end{defi}  

    

    

Intuitively, $\Exp$ assigns to each state a set of potential or expected observations. We now provide the model definition of the example mentioned in the introduction (cf. Figure \ref{figure:motivationalexamplekripkemodel}) For the sake of brevity, 
we do not draw
the reflexive arrows. If the $\vbot$ moves one step left, $\obsleft$, then while Alice still considers moving to the debris-disposal area a possibility, Bob does not consider that possibility at all, as described by Example \ref{ex-intro}, and depicted by the edge in Figure \ref{figure:motivationalexamplekripkemodel} between the states $u$ and $t$, annotated by Alice and not Bob.

\tikzstyle{world} = [draw]
\begin{figure}
	\begin{center}
		\begin{tikzpicture}[yscale=1.3]
			\node[world] (s) {$debris$};
			\node at (-0.1, -0.4) {\small $\expwater$};
			\node[world] (t) at (4, 0) {$power$};
			\node at (4.3, -0.4) {\small $\exppower$};
			\node[world] (u) at (0, 1) {$debris$};
			\node at (3.0, 1) {\small $\expwatere$};
			\node[left = 0mm of s] {$s$};
			\node[right = 0mm of t] {$t$};
			\node[left = 0mm of u] {$u$};
			\draw (s) edge node[above] {Alice, Bob} (t);
			\draw (s) edge node[left] {Alice} (u);
			\draw (t) edge node[above] {Alice} (u);
		\end{tikzpicture}
	\end{center}
	\vspace{-5mm}
\caption{Model describing the initial knowledge of the two agents Alice and Bob about the expectation of the $\vbot$.\label{figure:motivationalexamplekripkemodel}}
\end{figure}
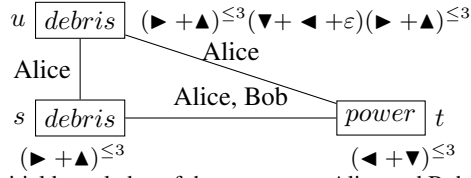

The logic $\POL$ was introduced to reason about agent knowledge via the matching of observations and expectations, and as we mentioned earlier, the difference between $\POL$ and $\POL^-$ is just a technical one. 
The main idea expressed in these logics is the following: While observing an action, people would tend to delete some impossible scenarios where they would not expect that observation to happen.
For this purpose, the update of epistemic expectation models with respect to some observation $w\in\Act^*$ is provided below. 

\begin{defi}[Update by observation]\label{def.upobs}
Let $w$ be an observation over $\Act$ and let $\M=\langle S,\sim,V,\Exp\rangle$ be an epistemic expectation model. The updated model $\M|_w= \langle S',\sim',V',\Exp'\rangle$ is defined by: $S' = \{s\mid \LL(\Exp(s)\regdiv w)\not=\emptyset\}$, ${\sim'_i}={\sim_i}|_{S'\times
	 S'},$ $V'=V|_{S'},$ and $\Exp'(s)=\Exp(s)\regdiv w$.
\end{defi}

The main idea of the updated model is 
to delete the states where the observation $w$ could not have happened. To reason about agent expectations and observations, the language for $\POL^-$ is provided below.

\begin{defi}[$\POL^-$ syntax]  Given a countable set of {\em propositional variables} $\BP$, a finite sets of {\em actions} $\Act$, and a finite set of {\em agents} $\Ag$, 
the {\em formulas} $\phi$ of $\POL^-$ are given by:
$$\begin{array}{r@{\quad::= \quad}l}
\phi  &
  \top
         \mid
           p
           \mid \neg \phi
           \mid \phi \land \phi
           \mid K_i\phi
           \mid [\pi] \phi
\end{array}$$
\noindent where $p\in\BP$, $i\in\Ag$, and $\pi\in\mathcal{L}_{\it obs}$. 
\end{defi}
 Intuitively, $K_i\phi$ says that `agent $i$ knows $\phi$ and $[\pi]\phi$ says that `after any observation in $\pi$, $\phi$ holds'. The other propositional connectives are defined in the usual manner. We also define $\ldiaarg{\pi}\phi$ as $\lnot [\pi]\lnot\phi$ and $\hat{K}_i\phi$ as $\lnot K_i \lnot\phi$. 
Typically, $\ldiaarg{\pi}\phi$ says that `there exists an observation in $\pi$ such that $\phi$ holds'. Formula $\hat{K}_i\phi$ says that `agent $i$ imagines a state in which $\phi$ holds'.

 The logic $\POL^-$ is the \textbf{$\starfree$ fragment of $\POL$}, that is, it is the set of formulas in which the $\pi$'s do not contain any Kleene star $*$. A more restricted version is the \textbf{$\word$ fragment of $\POLS$}, where $\pi$'s are words, that is,  observation expressions without $+$ operators. We consider both the \textbf{single-agent word fragment of $\POL^-$}, and \textbf{multi-agent word fragment of $\POL^-$}. Furthermore, we consider \textbf{single-agent $\POL^-$}, and \textbf{multi-agent $\POL^-$} (full $\POL^-$).   
 
\begin{defi}[Truth definition for $\POL^-$] Given an epistemic expectation model $\M$ = $(S, \sim ,V, \Exp)$, a state $s\in S$, and a $\POLS$-formula $\phi$, the truth of $\phi$ at $s$, denoted by $\M,s\vDash \phi$, is defined by induction on $\phi$ as follows: 
$$
\begin{array}{rcl}
\M,s\vDash p &\Leftrightarrow& p\in V(s)\\
\M,s\vDash \neg\phi &\Leftrightarrow&   \M,s\nvDash \phi \\
\M,s\vDash \phi\land \psi &\Leftrightarrow& \M,s\vDash \phi \textrm{ and } \M,  s\vDash \psi \\
\M,s\vDash K_i\phi &\Leftrightarrow& \textrm{for all }t: (s\sim_i t  \textrm{ implies } \M,t\vDash\phi)\\
\M,s\vDash [\pi]\phi &\Leftrightarrow& \textrm{for all observations }w\textrm{ over }\Sigma, \\
				&& w\in\LL(\pi) \cap \prefixes(\Exp(s))  \\
				&& \textrm{ implies }\M|_w,s\vDash \phi
\end{array}
$$

\noindent where $\prefixes(\pi)$ is the set of prefixes of words in $\LL(\pi)$, that is, $w\in \prefixes(\pi)$ 
iff $\exists v\in \Act^*$ such that $wv\in\LL(\pi)$ (namely $\LL(\pi\regdiv w)\not=\emptyset$).
\end{defi}
 
The truth of $K_i\phi$ at $s$ follows the standard possible world semantics of epistemic logic. The formula $[\pi]\phi$ holds at $s$ if for every observation $w$ in the set $\LL(\pi)$  that matches with the beginning of (i.e., is a prefix of) some expected observation in $s$, $\phi$ holds at $s$ in the updated model $\M|_w$. Note that $s$ is a state in $\M|_w$ because $w \in \prefixes(\Exp(s))$. Similarly, the truth definition of $\ldiaarg{\pi}\phi$ can be given as follows: $\M,s\vDash \ldiaarg{\pi}\phi \textrm{ iff there exists }w\in\LL(\pi) \cap \prefixes(\Exp(s)) 
 \textrm{ such that }\M|_w,s\vDash \phi$. Intuitively, the formula $\ldiaarg{\pi}\phi$ holds at $s$ if there is an observation $w$ in $\LL(\pi)$  that matches with the beginning of some expected observation in $s$, and $\phi$ holds at $s$ in the updated model $\M|_w$. For the example described earlier, we have:

\begin{itemize}
\item[-] $\M, t \models [\obsleft](K_{Bob} \neg debris \land \hat{K}_{Alice} debris)$, if the $\vbot$ moves one step left, $\obsleft$, then while Alice still considers moving to the debris-disposal area a possibility, Bob does not consider that possibility at all.
\end{itemize}


\ 

\noindent\textbf{Satisfiability problem for $\POL^-$: } 
Given a formula $\phi$, does there exist a pointed epistemic expectation model $\M,s$ such that  $\M,s \models \phi$? We investigate the complexity of this problem. The fragments of $\POL^-$ that we consider are (i) single-agent word fragment, (ii) multi-agent word fragment, (iii) single-agent $\POL^-$, and, (iv) full $\POL^-$.

\section{An application} 
\label{section:application}
Let us now consider a scenario which can be aptly described using the satisfiability problem of $\POLS$. We go back to the cleaning bot example introduced earlier. Let Alice be agent $a$ and Bob be agent $b$. Suppose the $\vbot$ is moving towards the power source without making any error. Evidently, the possibilities considered by the agents, based on the information available to them are given as follows: 
\begin{itemize}
    \item[-] Possibilities considered by Alice who has the information about the glitch in the bot:
    \begin{align*}
        \hat{K_a}debris\wedge \hat{K_a}\ldiaarg{\obsleft + \obsdown}debris\wedge\hat{K_a}power
    \end{align*}
    \item[-] Possibilities considered by Bob who is not aware of the glitch in the bot:
    \begin{align*}
        \hat{K_b}debris\wedge\hat{K_b}power
    \end{align*}
\end{itemize}

Now, we model the \emph{expectations} as follows: Consider the expression, $\pi^p_n = (\obsdown + \obsleft)^n$ that represents a sequence of moves of length $n$ the bot can make to get to to the power source without any error. We use a formula $P_n$ to express the following: As long as the bot is observed to make $n$ many moves towards the power source, reaching it is still a possibility.
    
    \begin{align*}
        P_n = &(\ldiaarg{\obsleft}\top\wedge\ldiaarg{\obsdown}\top)\\
        &\wedge[\pi^p_1](\ldiaarg{\obsleft}\top\wedge\ldiaarg{\obsdown}\top)\\
        &\wedge[\pi^p_2](\ldiaarg{\obsleft}\top\wedge\ldiaarg{\obsdown}\top)\ldots\\
        &\wedge[\pi^p_{n}](\ldiaarg{\obsleft}\top\wedge\ldiaarg{\obsdown}\top)
    \end{align*}
    The first conjunct of $P_n$ translates to move towards the power source, a move towards down or left can be observed. The second conjunct translates to the following: after the observation of a single left or down movement, another left or down movement can be observed. The other conjuncts can be described similarly.
    
    For the scenario described in the introduction, we can consider $P_n$ to create a formula where n is at most 3, without an error. 
    Let us denote such a formula by $\psi_{p}$. Similarly, a formula can express the movement towards debris-disposal with at most one error and with no error as $\psi_{de}$ and $\psi_d$, respectively.
    A situation where the bot is moving towards the power source without any error, but $a$ considers the possibility of moving towards debris-disposal with an error can be expressed as $\hat{K_a}\psi_{de}\wedge\psi_{p}$. Similarly, a formula can be considered for modelling the expected observation when both the agents consider the possibility of the bot moving towards debris-disposal area without an error: $\hat{K_a}\psi_{d}\wedge\hat{K_b}\psi_{d}$. 
We call the (finite) set of all such formulas, $\Gamma_{p}$. 
Similarly, we can construct a set $\Gamma_{de}$ of formulas, when the bot can make an error while going towards debris-disposal area or $\Gamma_{d}$ when it is moving towards the debris-disposal without any error. 

Suppose we want to conclude the following in the current scenario: After one wrong move, $b$ knows that the bot is not moving towards debris-disposal, but $a$ still considers the possibility. The formula, $\mathit{INFO}_{ab}$, say, turns out to be
$$
\ldiaarg{\obsdown + \obsleft}(K_b power\wedge \hat{K_a}debris)
$$
The actual scenario is that the bot is indeed moving towards $power$. Hence, to check whether $\mathit{INFO}_{ab}$ can be concluded in this scenario,   a satisfiability solver for $\POL^-$ can check the (un)satisfiability of the formula
$$
\neg{((\bigwedge_{\psi\in\Gamma_{p}}\psi)\rightarrow \mathit{INFO}_{ab})}
$$


\section{Algorithm for the Satisfiability Problem of \texorpdfstring{$\POLS$}{POLSatAlgo}}\label{section:starfreetableau}
In this section, we design a proof system using the tableau method to prove satisfiability of $\POLS$.  

A term in a tableau proof is of the form $(\sigma\ \ w\ \ \psi)\mid (\sigma\ \ w\ \ \checkmark)\mid (\sigma,\sigma')_i$, where $i\in Agt$. The $\sigma$ is called a state label that represents a state in the model, $w\in\Sigma^*$ is a word over a finite alphabet and $\psi$ is a formula in $\POLS$.

The term $(\sigma\ \ w\ \ \psi)$ represents the fact that the state labelled by $\sigma$ survives after the model is projected on the word $w$, and after projecting on $w$, $\psi$ holds true in the state corresponding to $\sigma$.

The term $(\sigma\ \ w\ \ \checkmark)$ represents the fact that the state labelled by $\sigma$ survives after the model is projected on word $w$.

The term $(\sigma_1,\sigma_2)_i$ represents in the model, the states represented by $\sigma_1$ and $\sigma_2$ should be indistinguishable for the agent $i\in Agt$, where $Agt$ is a finite set of agents.

For space reasons, the term $(\sigma_1, \sigma_2)_{i\in Agt}$ stands for the set of terms $\{(\sigma_1,\sigma_2)_i\mid i\in Agt\}$.

Without loss of generality, the formula $\varphi$ is assumed to be in Negative Normal form, the syntax of which is as follows:
\begin{align*}
    \varphi := & \top\ \  |\ \  p\ \  |\ \  \neg p\ \  |\ \  \psi \vee \chi\ \ |\ \  \psi \wedge \chi\ |\ \\
    &  \ \hat{K_i}\psi\ \ |\ \ K_i\psi\ \ |\ \ \ldiaarg\pi\psi\ \ |\ \ [\pi]\psi
\end{align*}

Given a formula we denote by $\varphi$, $FL(\varphi)$  the Fischer-Ladner Closure of $\varphi$, (see \cite{DL}).
\subsection{The Tableau Rules}
The tableau rules for this fragment have been shown in Figure \ref{fig:tableaurules}. Here an inference rule looks like this:
\infer{C_1 | C_2 | \ldots | C_n}{%
                        A
                        }.

Here each $C_i$ and $A$ is a set of tableau terms. The $C_i$s are called consequences, $A$ is the antecedent. Intuitively the rule is interpreted as "If all the terms in  $A$ are true, then all the terms in at least one of $C_i$'s are true".


\newcommand{\tableaurulesection}[1]{\\ \multicolumn{2}{c}{\textbf{#1}} \\}
\newcommand{\tableaurule}[2]{\raisebox{2.5mm}{#1} & #2 \\[2mm]}

\begin{figure}[t!]
	\scalebox{0.8}{
\begin{tabular}{p{4cm}c}
    \multicolumn{2}{c}{\textbf{Propositional Rules}}\\
    \tableaurule{Clash rule}{
        \infer{\bot}{%
                        (\sigma
                        & w
                        & p),
                        & (\sigma
                        & w
                        & \neg p)
                        }
                    }
        
        \tableaurule{ AND rule}
        {\infer{(\sigma\ \ w\ \ \psi),(\sigma\ \ w\ \  \chi)}{%
        		(\sigma
        		& w
        		& \psi\wedge\chi)
        }}
        
        \tableaurule{ OR rule}
        {\infer{(\sigma\ \ w\ \ \psi)\ \ |\ \ (\sigma\ \ w\ \ \chi)}{%
        		(\sigma
        		& w
        		& \psi\vee\chi)
        }}

    \tableaurulesection{Knowledge Rules}

      \tableaurule{Knowledge}
      { 
    \infer{(\sigma'\ \ w\ \ \psi)}{%
    	(\sigma
    	& w
    	& K_i\psi),
    	&(\sigma'
    	& w
    	& \checkmark),
    	&(\sigma, \sigma')_i
    }
}
    
    \tableaurule{Possibility}
    {
    \infer{(\sigma,\sigma_n)_i, (\sigma_n\ \ w\ \ \checkmark), (\sigma_n\ \ w\ \  \psi), (\sigma_n, \sigma_n)_{i\in Agt}}{%
    	(\sigma
    	& w
    	& \hat{K_i}\psi)
    }}

 \tableaurule{Transitivity}{
    \infer{(\sigma,\sigma')_i}{%
    	(\sigma, \sigma'')_i,
    	& (\sigma'',\sigma')_i
    }}
 \tableaurule{Symetry}{
    \infer{(\sigma,\sigma')_i}{%
    	(\sigma', \sigma)_i
    }, $i\in Agt$}

    \tableaurulesection{Diamond and Box Rules}
    
     \tableaurule{Diamond Decompose}{
    \infer{(\sigma\ \ w\ \ \langle \pi\rangle\langle \pi'\rangle\psi)}{%
    	(\sigma
    	& w
    	& \langle \pi\pi'\rangle\psi)
    }}
    
    
    \tableaurule{ Diamond ND Decompose}{
    \infer{(\sigma\ \ w\ \ \ldiaarg{\pi_1}\psi)\ \ |\ \ (\sigma\ \ w\ \ \ldiaarg{\pi_2}\psi)}{%
    	(\sigma
    	& w
    	& \ldiaarg{\pi_1 + \pi_2}\psi)
    }}
    
    
    \tableaurule{ Diamond Project}{
    \infer{(\sigma\ \ wa\ \ \checkmark), (\sigma\ \ wa\ \ \psi)}{%
    	(\sigma
    	& w
    	& \langle a\rangle\psi)
    }}
    
    \tableaurule{ Box Project}{
    \infer{(\sigma\ \ wa\ \ [\pi\regdiv a]\psi)}{%
    	(\sigma
    	& w
    	& [\pi]\psi),
    	& (\sigma\ \ wa\ \ \checkmark)
    }}
    
    \tableaurule{ Empty Box}{
    \infer{(\sigma\ \ w\ \ \psi)}{%
    	(\sigma
    	& w
    	& [\epsilon]\psi)
    }}
    

    \tableaurulesection{Survival Rules}

    \tableaurule{Constant Valuation Up}{
    \infer{(\sigma\ \ \epsilon\ \ p)}{%
    	(\sigma
    	& w
    	& p)
    }~~~~~
    \infer{(\sigma\ \ \epsilon\ \ \neg p)}{%
    	(\sigma
    	& w
    	& \neg p)
    }
}
    
    \tableaurule{Survival Chain}
    {    \infer{(\sigma\ \ w\ \ \checkmark)}{%
    	(\sigma
    	& wa
    	& \checkmark)
    }
}

\end{tabular}
}
\caption{Tableau rules. $\sigma$ is any state symbol, $w$ is any word, $p$ is any propositional variable, $i$ is any agent, $\pi$ is any regular expression, $a$ is any letter.}\label{fig:tableaurules}
\end{figure}

In Figure \ref{fig:tableaurules}, the left column is the rule name and the right column is the rule. For example, the Box Project Rule states that "The state labelled by $\sigma$ survives after projection on word $w$ and it satisfies $[\pi]\psi$ ($(\sigma\ \ w\ \ [\pi]\psi)$) and $\sigma$ still survives a further projection on letter $a$($(\sigma\ \ wa\ \ \checkmark)$) then after further projection on $a$, $[\pi\regdiv a]\psi$ should hold true in the state labelled by $\sigma$ ($(\sigma\ \ wa\ \ [\pi\regdiv a]\psi)$).". Recall $\pi\regdiv a$ denotes the residual of $\pi$ by $a$ (see Section \ref{section:background}). 

Similarly, the Diamond Project rule says that if a certain state $\sigma$, under some word projection $w$ has to satisfy $\ldiaarg{a}\psi$, then that state $\sigma$ has to survive projection on $wa$ and also satisfy $\psi$ under the same projection.

A tableau proof can be assumed a tree. Each node of the tree is a set of tableau terms $\Gamma$. An inference rule can be applied in the following way:

If $A\subseteq\Gamma$ and $C_i$'s are not in $\Gamma$, the children of $\Gamma$ are $\Gamma\union C_i$ for each $i\in[n]$.

When no rules can be applied on a $\Gamma$, we say $\Gamma$ is saturated (leaf node in the proof tree).

If $\bot\in\Gamma$, we say that branch is \textbf{closed}. If all branch of the proof tree is \textbf{closed}, we say the \textbf{tableau is closed}, else is open.

Given a $\POLS$ formula $\varphi$, we start with $\Gamma = \{(\sigma\ \ \epsilon\ \ \varphi), (\sigma\ \ \epsilon\ \ \checkmark)\}\union \set{(\sigma,\sigma)_i, i\in Agt}$.

\begin{example}Suppose we aim at deciding whether $$\phi := \hat{K_i}\ldiaarg{a}p\wedge \ldiaarg{a}K_i \lnot p$$ is satisfiable or not. For simplicity we suppose there is a single agent $i$. Here are the terms added to the set of terms:

\begin{enumerate}
    \item $(\sigma\ \ \epsilon\ \ \varphi)$, $(\sigma\ \ \epsilon\ \ \checkmark)$,  $(\sigma,\sigma)_i$\hfill (initialization) 
    \item $(\sigma\ \ \epsilon\ \ \hat{K_i}\ldiaarg{a}p)$, $(\sigma\ \ \epsilon\ \ \ldiaarg{a}K_i \lnot p$) \hfill by AND rule
    \item $(\sigma'\ \ \epsilon\ \ \ldiaarg{a}p), (\sigma'\ \ \epsilon\ \ \checkmark), (\sigma,\sigma')_i, (\sigma', \sigma')_i$\hfill by Possibility rule
    \item $(\sigma', \sigma)_i$\hfill by Symmetry rule
    \item $(\sigma'\ \ a\ \ p), (\sigma'\ \ a\ \ \checkmark)$\hfill by Diamond Project on 2
    \item $(\sigma\ \ a\ \ \checkmark), (\sigma\ \ a\ \ K_i\lnot p)$\hfill by Diamond Project on 2
    \item $(\sigma'\ \ a\ \ \lnot p)$\hfill by Knowledge rule on 3, 5, 6
    \item $\bot$ \hfill by Clash rule on 5,7
\end{enumerate}

As we obtain $\bot$, the formula $\phi$ is not satisfiable (by the upcoming Theorem~\ref{thm:starfreemultisound}).
\end{example}

\subsection{Soundness and Completeness of the Tableau Rules}
In this section, we provide the soundness and completeness proof of the Tableau method for the satisfiability of $\POLS$
\begin{theorem}\label{thm:starfreemultisound}
Given a formula $\varphi$, if $\varphi$ is satisfiable, then the tableau for $\Gamma = \{(\sigma\ \ \epsilon\ \ \varphi), (\sigma\ \ \epsilon\ \ \checkmark), (\sigma,\sigma)_{i\in Agt}\}$ is open.
\end{theorem}

\begin{theorem}\label{thm:completeness}
Given a formula $\varphi$, if the tableau for $\Gamma = \{(\sigma\ \ \epsilon\ \ \varphi), (\sigma\ \ \epsilon\ \ \checkmark), (\sigma,\sigma)_{i\in Agt}\}$ is open, then $\varphi$ is satisfiable.
\end{theorem}

The proof of Theorem~\ref{thm:starfreemultisound} is done by induction. We shift the proof of Theorem~\ref{thm:starfreemultisound} to the appendix. We now present the proof of Theorem~\ref{thm:completeness}.

\begin{proof}[Proof of Theorem~\ref{thm:completeness}]
Since by assumption, the tableau for $\Gamma = \{(\sigma\ \ \epsilon\ \ \varphi), (\sigma\ \ \epsilon\ \ \checkmark), (\sigma,\sigma)_{i\in Agt}\}$ is open, there exists a branch in the tableau tree where in the leaf node there is a set of terms $\Gamma_l$ such that it is saturated and $\bot\notin\Gamma_l$.

For the purpose of this proof, let us define a relation over the words $\bar{w}$ that appears in $\Gamma_l$. For any two word $\bar{w}_1$ and $\bar{w}_2$ that appears in $\Gamma_l$, $\bar{w}_1\leq_{pre}\bar{w}_2$ if and only if $\bar{w}_1\in\prefixes(\bar{w}_2))$. Now, this relation is reflexive ($\bar{w}_1\in\prefixes(\bar{w}_1)$), asymmetric (if $\bar{w}_1\in\prefixes(\bar{w}_2)$ and $\bar{w}_2\in\prefixes(\bar{w}_1)$ then $\bar{w}_1 = \bar{w}_2$) and transitive (if $\bar{w}_1\in\prefixes(\bar{w}_2)$ and $\bar{w}_2\in\prefixes(\bar{w}_3)$ then $\bar{w}_1\in\prefixes(\bar{w}_3)$). Hence this relation creates a partial order among all the words occurring in $\Gamma_l$. We also denote $\bar{w_1}<_{pre}\bar{w_2}$ to interpret the fact that $\bar{w_1}\leq_{pre}\bar{w_2}$ and $\bar{w_1}\neq\bar{w_2}$.

Now we create a model $\M = \ldiaarg{W, \{R_i\}_{i\in Agt}, V, Exp}$ out of $\Gamma_l$ and prove that $\varphi$ is satisfied by some state in the model.

\begin{itemize}
    \item $W = \{s_\sigma\mid \sigma\mbox{ is a distinct label in the terms occuring in }\Gamma_l\}$
    
    \item $R_i = \{\{s_{\sigma_1}, s_{\sigma_2}\}\mid(\sigma_1,\sigma_2)_i\in\Gamma_l\}$
    
    \item $V(s_\sigma) = \{p\mid (\sigma\ \ \epsilon\ \ p)\in\Gamma_l\}$
    
    \item $Exp(s_\sigma) = \sum_{w\in\Lambda_\sigma}w$, where $\Lambda_\sigma = \{w\mid (\sigma\ \ w\ \ \checkmark)\in\Gamma_l\mbox{ and }\nexists w':((\sigma\ \ w'\ \ \checkmark)\in\Gamma_l\mbox{ and }w<_{pre} w')\}$
\end{itemize}

Note that, the new state label $\sigma_n$ is only created in the possibility rule, with a reflexive relation on itself. Now consider the set $R' = \{(\sigma, \sigma')\mid \{(\sigma\ \ w\ \ \checkmark), (\sigma' \ \ w'\ \ \checkmark)\}\subseteq\Gamma_l\}$. Hence this can be considered a binary relation over the set of all distinct $\sigma$ that occurs in $\Gamma_l$. When a $\sigma'$ is created by the possibility rule, it is reflexive. Also by the relation rules, they are made symmetrically and transitively related to every other label that has been previously there. Hence $R'$ is an equivalence relation, hence making $R_i$ in the model an equivalence relation.

Now, Theorem~\ref{thm:completeness} follows from the following two claims, the proofs of which we present later.

\begin{claim} \label{claim:stm1}
If $(\sigma\ \ w\ \ \checkmark)\in\Gamma_l$ then $s_{\sigma}$ survives in $\M|_w$.
\end{claim}

\begin{claim} \label{claim:stm2}  For any word $w$ that occurs in $\Gamma_l$, any label $\sigma$ and any formula $\psi$, If $(\sigma\ \ w\ \ \psi)\in\Gamma_l$ and $(\sigma\ \ w\ \ \checkmark)\in\Gamma_l$ then $s_{\sigma}$ survives in $\M|_w$ and  $\M|_w,s_\sigma\vDash\psi$.
\end{claim}

\begin{proof}[Proof of Claim~\ref{claim:stm1}]
        We induct on the size of $|w|$.

        \textbf{Base Case.} Let $|w| = 1$. Hence $w \in \{\epsilon\}\union \Sigma$ . Since $\Gamma\subseteq\Gamma_l$ and $(\sigma\ \ \epsilon\ \ \checkmark)$, and $s_\sigma$ is in $\M|_\epsilon = \M$.
        
        For the case $w = a$ for any $a\in\Sigma$. Hence there exists a word $w'$ that occurs in a term in $\Gamma_l$ labelled by $\sigma$ such that $w\in\prefixes(w'))$ and there is no other word bigger than $w'$ such that $w'$ is in its prefix, since the proof is on finite words and formula, the proof terminates. Hence by definition of $w'\in\LL(Exp(s_\sigma))$ which guarantees survival of $s_\sigma$ in $\M|_a$. 
        
        \textbf{Induction Hypothesis.} Assume the statement to be true for $|w| = n$.
        
        \textbf{Inductive Step.} Consider the case where $|w| = n + 1$. 
        
        By assumption, $(\sigma\ \ w\ \ \checkmark)\in\Gamma_l$. Hence by the fact that $\Gamma_l$ is saturation and by the rule "Survival Chain", there is $(\sigma\ \ w'\ \ \checkmark)\in\Gamma_l$, where $w = w'a$ for some $a\in\Sigma$. Hence by IH, the result follows that $s_\sigma$ survives in $\M|_{w'}$. 
        
        Now, by termination, there are finite many unique words occurring in $\Gamma_l$. Clearly, $w'\leq_{pre} w$. Since there are finite many words, there is a $w_*$, which is of maximum size such that $w\leq_{pre} w_*$ and $(\sigma\ \ w_*\ \ \checkmark)\in\Gamma_l$. Hence $w_*\in\Lambda_\sigma$ in the definition of Exp of the model. Therefore $w_*\in\LL(Exp(s_\sigma))$ and since $w'\leq_{pre} w\leq_{pre} w_*$, $s_\sigma$ survives in $\M|_{w'}$, hence $s_\sigma$ shall survive in $\M|_{w}$. \qedhere
\end{proof}

\begin{proof}[Proof of Claim~\ref{claim:stm2}]
        Naturally, we shall induct upon the size of $\psi$. 

        \textbf{Base Case.} Let $\psi$ is of the form $p$ or $\neg p$. By the definition of the function $V$ for the model and the previous proof, the statement stands true.
        
        \textbf{Induction Hypothesis.} Let us consider the statement is true for any $\psi$ such that $|\psi|<n'$ for some $n'$.
        
        \textbf{Inductive Step.} We prove for $|\psi| = n'$. Again, we go case by case on the syntax of $\psi$. 
        \begin{itemize}
            \item $\psi = \hat{K_i}\chi$. Since $\Gamma_l$ is saturated, by the rule of possibility, $\{(\sigma'\ \ w\ \ \chi), (\sigma, \sigma')_i, (\sigma'\ \ w\ \ \checkmark)\}\subseteq\Gamma_l$. By IH on the subformula $\chi$, the definition of the model, the proof of the previous statement, and the rule "survival chain", $s_{\sigma'}$ survives in $\M|_{w}$ and $\M|_w,s_{\sigma'}\vDash\chi$. Also by definition, $\{s_{\sigma}, s_{\sigma'}\}\in R_i$, hence proving $\M|_w, s_\sigma\vDash\hat{K_i}\chi$.
            
            \item $\psi = K_i\chi$. Since $\Gamma_l$ is saturated, and by previous statement $s_{\sigma'}$ is surviving for every $(\sigma'\ \ w\ \ \checkmark)$, by the rule of knowledge $(\sigma'\ \ w\ \ \chi)\in\Gamma_l$ for every $(\sigma, \sigma')_i$. Hence by IH on subformula, $\M|_w, s_{\sigma'}\vDash\chi$ for every $\sigma'$ such that $\{\sigma,\sigma'\}\in R_i$.
        
            \item $\psi = \ldiaarg{\pi + \pi'}\chi$. Since $\Gamma_l$ is saturated, hence by the ND Decomposition, either the term $(\sigma\ \ w\ \ \ldiaarg{\pi}\chi)\in\Gamma_l$ or $(\sigma\ \ w\ \ \ldiaarg{\pi'}\chi)\in\Gamma_l$. By IH, $\M|_w,s_\sigma\vDash\ldiaarg{\pi}\chi$ or $\M|_w,s_\sigma\vDash\ldiaarg{\pi'}\chi$ and hence $\M|_w,s_\sigma\vDash\ldiaarg{\pi + \pi'}\chi$.
        
            \item $\psi = \ldiaarg{\pi\pi'}\chi$. Since $\Gamma_l$ is saturated, hence $(\sigma\ \ w\ \ \ldiaarg{\pi}\ldiaarg{\pi'}\chi)\in \Gamma_l$. By IH, since $\ldiaarg{\pi}\ldiaarg{\pi'}\chi\in FL(\psi)$, hence $\M|_w, s_\sigma\vDash\psi$.
            
            \item $\psi = \ldiaarg{a}\chi$. Note that we don't consider a general word $w'$ in the diamond as given $w' = aw''$, a formula $\ldiaarg{w'}\chi$ is satisfiable if and only if $\ldiaarg{a}\ldiaarg{w''}\chi$ is satisfiable.
        
            \item $\psi = [\pi]\chi$. Let us consider $(\sigma\ \ wa\ \ \checkmark)\in\Gamma_l$ for some $a\in\Sigma$. Hence by the proof of the first statement, $s_\sigma\in\M|_{wa}$. Also $|\LL(\pi)| < |\LL(\pi\regdiv a)|$. Hence by induction on the size of formula $\M|_{wa},s\sigma\vDash[\pi\regdiv a]\chi$ which implies $\M|_w,s_\sigma\vDash[\pi]\chi$.  \qedhere
        \end{itemize}
        \end{proof}
    This completes the proof of Theorem~\ref{thm:completeness}
    \end{proof}

\subsection{A \texorpdfstring{$\NEXPTIME$}{nexptimeub} Upper Bound}
Now we design an algorithm based on tableau and prove existence of an algorithm that takes non-deterministically exponential steps with respect to the size of $\varphi$. Now given a $\varphi$, we now create a tree of nodes,  where each node $T_\sigma$ contains terms of the tableau of the form $(\sigma\ \ w\ \ \psi)$ and $(\sigma\ \ w\ \ \checkmark)$, where $w\in \Sigma^*$ is a word that is occuring in tableau, and $\psi$ is a formula in $FL(\varphi)$. Each node $T_\sigma$ refers to a state label $\sigma$ in tableau, a term of the $(\sigma\ \ w\ \ \psi)\in T_\sigma$ intuitively translates to in the state corresponding to $\sigma$, after projecting model on $w$, the state survives and there $\psi$ is satisfied. Similarly, $(\sigma\ \ w\ \ \checkmark)\in T_\sigma$ means state corresponding to $\sigma$ survives after projection on $w$. The tableau tree created, we call it $\cT_\tP$


We saturate the rules carefully such that each node in the tree corresponds to a single state in the model. This technique is well studied in \cite{DBLP:journals/ai/HalpernM92}.

\begin{theorem}\label{thm:NEXPTIME}
    The satisfiability of $\POLS$ is in $\NEXPTIME$.
\end{theorem}

\begin{proof}
    Given the tree $\cT_\tP$ we create in the procedure, a node $T_\sigma$ is marked satisfiable iff it does not have bot, $\{(\sigma\ \ w\ \ K_i\psi), (\sigma\ \ w\ \ \neg\psi)\}\nsubseteq T_\sigma$ and all its successors are marked satisfiable. We prove three statements:
    \begin{itemize}
        \item \textbf{Statement 1:} Each node is of at most exponential size, that is, has at most exponential many terms.
        \item \textbf{Statement 2:} Maximum children a node can have is polynomial.
        \item \textbf{Statement 3: } The height of the tree is polynomial.
    \end{itemize}

    \noindent\textbf{Proof of Statement $1$}. Since a term in a node $T_\sigma$ is of the form $(\sigma\ \ w\ \ \psi)$, where $w$ is a word over some finite alphabet $\Sigma$ and $\psi$ is a formula of $\POLS$. 

   According to the shape of the rules, a formula that can be derived is always  in $ FL(\varphi)$. Since $|FL(\varphi)|\leq O(|\varphi|)$\cite{DL}, hence there can be at most $O(|\varphi|)$ many formulas.

    Also, since a regular expression $\pi$ occuring in a modality is star-free (that is does not contain the Kleene star), hence a word $w\in\LL(\pi)$ is of length at most $|\pi|$ which is again of length at most $\varphi$. Also there are at most $|FL(\varphi)|$ many regular expressions. Hence there are at most $|\Sigma|^{O(p(|\varphi|))}$, where $p(X)$ is some polynomial on $X$, many unique words possible. Hence therefore, there can be at most exponential many terms in a single node.

    \noindent\textbf{Proof of Statement $2$}. From a node $T_\sigma$, a child is created for every unique triplet of $(\sigma\ \ w\ \ \hat{K}\psi)$ in $T_\sigma$. Number of such triplets possible is, as proved is at most polynomial with respect to $|\varphi|$.

    \noindent\textbf{Proof of Statement $3$}. For proving this, we use $md(\Gamma)$, given a set of formulas $\Gamma$, is the maximum modal depth over all formulas in $\Gamma$. Finally we define $F(T_\sigma)$ as the set of formulas occuring in the node $T_\sigma$. 

    Consider $T_\sigma$, the node $T^i_{\sigma'}$ is $i$- successor of $T_\sigma$ and $T^j_{\sigma''}$ be the $j$ successor of $T^i_{\sigma'}$ ($i\neq j$). Note that all the formulas in $F(T^j_{\sigma''})$ are from FL closure of all the $K_j$ and $\hat{K_j}$ formulas from $F(T^i_{\sigma'})$.

    Also all the formulas in $F(T^i_{\sigma'})$ are in the FL closure of the $K_i$ and $\hat{K_i}$ formulas occusring in $T_\sigma$. Hence $md(T^j_{\sigma''})\leq md(F(T^i_{\sigma'}))$. Therefore, there can be at most $O(|\varphi|^c)$ such agent alterations in one path of $\cT_P$ (not linear because there can be polynomial many words paired with each formula).


    Now let us consider how many consecutive $i$ succesors can happen in a path. Suppose a $T_{\sigma}$ has a new $i$-successor node $T_{\sigma'}$ for the term $(\sigma\ \ w\ \ \hat{K_i}\psi)$. Due to the fact that the indistinguishability relation is equivalence for each agent due to the Transitivity, Symmetry rule and the reflexivity that infers in the possibility rule, hence all the possibility and the knowledge formula terms of the form $(\sigma\ \ w'\ \ \hat{K_i}\xi)$ or $(\sigma\ \ w'\ \ K_i\xi)$ of agent $i$ are in the successor node $T_{\sigma'}$ in the form $(\sigma'\ \ w'\ \ \hat{K_i}\xi)$ or $(\sigma'\ \ w'\ \ \hat{K_i}\xi)$ respectively, along with the term $(\sigma'\ \ w \ \ \psi)$. Hence the number of such unique combination of terms will be at most polynomial to the size of $|FL(\varphi)|$. 

    Therefore, the height of $\cT_{\tP}$ is polynomial with respect to the $|\varphi|$.
\end{proof}

\section{Hardness of Satisfiability in \texorpdfstring{$\POLS$}{POLSatHard}}\label{section:nexptimehard}
\newcommand{\lboxarg}[1]{\left[#1\right]}
\newcommand{\tilingprop}[2]{q_{#1}^{#2}}
\newcommand{\lbox}{\square}
\newcommand{\ldia}{\Diamond}
\newcommand{\limply}{\rightarrow}
\newcommand{\bottom}{\bot}
\newcommand{\binarytree}{\mathcal T}
In this section, we give a lower bound to the Satisfiability problem of $\POLS$. We reduce the well-known $\NEXPTIME$-Complete Tiling problem to come up with a formula in the $\POLS$ fragment that only has $2$ agents.

\begin{figure}[t!]
	\begin{center}
		\begin{tikzpicture}[scale=0.5]
		
		\node[text width=2cm] at (-2, 1.5) 
		{					
			\tilewcsmall \tilewhite \tilewhite \tilegreen\tilered
			
			\tilewcsmall \tilegreen\tilewhite\tilegreen\tilered
			
			\tilewcsmall \tilegreen\tilewhite\tilewhite\tilewhite
			
			\tilewcsmall \tilewhite \tilered\tilewhite\tilered
		};
		
		\node[text width=2cm] at (-0.5, 1.5) {
			\tilewcsmall \tilegreen \tilewhite\tilegreen\tilewhite
			
			\tilewcsmall \tilewhite \tilewhite\tilewhite\tilegreen
			
			\tilewcsmall \tilewhite\tilered\tilegreen\tilewhite
			
			\tilewcsmall \tilegreen\tilered\tilered\tilewhite
			
			\tilewcsmall \tilered	\tilegreen\tilewhite\tilewhite
			
		};
		
		\draw[fill=lightgray] (0, 4) rectangle (4, 0);
		\foreach \x in {0, 1, 2, 3} {
			\draw[dashed] (\x, 0) -- (\x, 4);	
		}
		\foreach \y in {0, 1, 2, 3} {
			\draw[dashed] (0, \y) -- (4, \y);
		}

  \tile 0 {0} \tilewhite\tilered\tilegreen\tilewhite
		
		\end{tikzpicture}
		\hspace{20pt}
		\begin{tikzpicture}[scale=0.5, baseline=-15mm]
        \tile 0 1 \tilewhite \tilewhite \tilegreen\tilewhite
		\tile 1 1 \tilegreen\tilewhite\tilegreen\tilewhite
		\tile 2 1 \tilegreen\tilewhite\tilegreen\tilewhite
        \tile 3 1 \tilegreen\tilewhite\tilered\tilered
		\tile 0 0 \tilewhite \tilewhite \tilegreen\tilered
		\tile 1 0 \tilegreen\tilewhite\tilegreen\tilered
		\tile 2 0 \tilegreen\tilewhite\tilewhite\tilewhite
        \tile 3 0 \tilewhite\tilered\tilewhite\tilered
		\tile 0 {-1} \tilewhite \tilered\tilewhite\tilered
		\tile 1 {-1} \tilewhite \tilered\tilewhite\tilered
		\tile 2 {-1} \tilewhite \tilewhite\tilewhite\tilegreen
        \tile 3 {-1} \tilewhite\tilered\tilewhite\tilered
		\tile 0 {-2} \tilewhite\tilered\tilegreen\tilewhite
		\tile 1 {-2} \tilegreen\tilered\tilered\tilewhite
		\tile 2 {-2} \tilered	\tilegreen\tilewhite\tilewhite
        \tile 3 {-2} \tilewhite\tilered\tilewhite\tilered
		\end{tikzpicture}
	\end{center}
	\caption{A set of tile types and an empty square, and a solution.\label{figure:tiling}}
\end{figure}
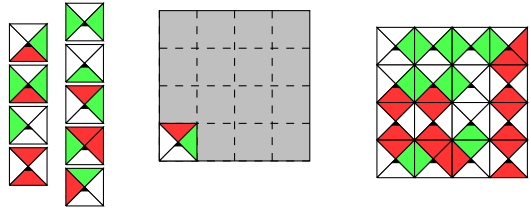

\begin{theorem}\label{theorem-nexptimehard}
$\POLS$ satisfiability problem is $\NEXPTIME$-Hard.
\end{theorem}
\begin{proof}
We reduce the $\NEXPTIME$-Complete tiling problem of a square whose size is $2^n$ where $n$ is encoded in unary \cite{van2019convenience} (see Figure~\ref{figure:tiling}). The instance of the tiling problem is $(\tileset, \tileseed, n)$ where $\tileset$ is a set of tile types (e.g \scalebox{0.5}{\tilewcsmall \tilegreen \tilewhite\tilegreen\tilewhite}), $\tileseed$ is a specific tile that should be at position $(0, 0)$, and $n$ is an integer given in unary. Note that the size of the square is exponential in $n$. We require the colours of the tiles to match horizontally and vertically.

The idea of the reduction works as follows. We consider two tilings A and B. We will construct a formula $tr(T, t_0, n)$ expressing that the two tilings are equal, contains $t_0$ at (0, 0), and respect the horizontal and vertical constraints. 

\newcommand{\lev}[1]{lev_{#1}}

With the help of two epistemic modalities $K_i$ and $K_j$ we can simulate a standard $K$ modal logic $\lbox$. 
For the rest of the proof, we consider such a $\lbox$ modality and its dual $\Diamond$. We encode a binary tree whose leaves are pairs of positions (one position in tiling A and one in tiling B). Such a tree is of depth $4n$: $n$ bits to encode the $x$-coordinate in tiling A, $n$ bits to encode the $x$-coordinate in tiling B, $n$ bits to encode the $y$-coordinate in tiling A, $n$ bits to encode the $y$-coordinate in tiling B. A pair of positions is encoded with the $4n$ propositional variables: $p_0, \dots, p_{4n-1}$. The first $p_0, \dots, p_{2n-1}$ encodes the position in tiling $A$ while the later $p_{2n}, \dots, p_{4n-1}$ encodes the position in tiling $B$. At each leaf, we also use propositional variables $\tilingprop t A$ (resp. $\tilingprop t B$) to say there is tile $t$ at the corresponding position in tiling $A$ (resp. tiling $B$). The following formula enforces the existence of that binary tree $\binarytree$ by branching over the truth value of proposition $p_\ell$ at depth $\ell$:
\begin{equation}
\bigwedge_{\ell<4n}\!\!\!\lbox^\ell\!\! \left(\!\!\ldia p_\ell \land \ldia \lnot p_\ell \land \bigwedge_{i<\ell} (p_i {\limply} \lbox p_i) \land (\lnot p_i {\limply} \lbox \lnot p_i)\!\!\right)
\end{equation}
Now, by using of specific Boolean formulas over $p_0, \dots, p_{4n-1}$, it is easy to express equality, presence of $t_0$ at $(0, 0)$ and horizontal and vertical constraints:

\begin{align}
\lbox^{4n} \left(\bigvee_{t} \tilingprop t A ~~ \land ~~ \bigwedge_{t \neq t'} (\lnot\tilingprop t A \lor \lnot\tilingprop {t'} A) \right)\\
\lbox^{4n} \left(\bigvee_{t} \tilingprop t B ~~ \land ~~ \bigwedge_{t \neq t'} (\lnot\tilingprop t B \lor \lnot\tilingprop {t'} B) \right)\\
\lbox^{4n} (\text{position in tiling $A$ = 0}) \limply \tilingprop {t_0} A \\
\lbox^{4n} (\text{\begin{tabular}{l}$x$-coordinate of position in $A$ \\ = 1 + $x$-coordinate of position in $B$\end{tabular}})  \\
~~~~~~~~~~~~~~~~~\limply  \bigvee_{t ,  t' \mid \text{$t$ matches $t'$ horizontally}} (\tilingprop t A \land \tilingprop {t'} B) \\
\lbox^{4n} (\text{\begin{tabular}{l}$y$-coordinate of position in $A$ \\ = 1 + $y$-coordinate of position in $B$\end{tabular}})  \\
~~~~~~~~~~~~~~~~~\limply  \bigvee_{t ,  t' \mid \text{$t$ matches $t'$ vertically}} (\tilingprop t A \land \tilingprop {t'} B)
\end{align}
The main difficulty is to be sure that all pairs of positions with the same position for - let's say - tiling $A$ indicates the same tile for the tiling $A$ (i.e. the same variable $\tilingprop t A$ is true). To this aim, we will write a formula of the following form
$$[\pi_{\text{
any position  in A
}}] \bigvee_t \lbox^{4n} \tilingprop t A ~\land~ [\pi_{\text{any position in B}}] \bigvee_t \lbox^{4n} \tilingprop t B.$$
To be able to perform observations to select any position in tiling A (resp. B) whatever the position in tiling B (resp. A) is, we introduce the alphabet $\Sigma = \set{A, \bar A, B, \bar B}$. We write these two formulas that make a correspondence between valuations on the leaves and observations:
\begin{equation}
    \lbox^{4n}\bigwedge_{i=0..2n-1} \!\!\!\!\!\! [A + \bar A]^i \left(\begin{array}{ll}(p_i \limply \ldiaarg{A}\top \land [\bar A]\bottom) \land \\ (\lnot p_i \limply \ldiaarg{\bar A} \top \land [A]\bottom)\end{array}\right)
\end{equation}
\begin{equation}\lbox^{4n} \!\!\!\!\!\!\bigwedge_{i=2n..4n-1} \!\!\!\!\!\! [B + \bar B]^{i-2n} \left(\begin{array}{ll}(p_i \limply \ldiaarg{B}\top \land [\bar B]\bottom) \land \\ (\lnot p_i \limply \ldiaarg{\bar B} \top \land [B]\bottom)\end{array}\right)
\end{equation}
The idea is that a $2n$-length word on alphabet $\set{A, \bar A}$ corresponds to a valuation over $p_1, \dots, p_{2n-1}$, and thus a position in tiling A and only that $2n$-length word on alphabet $\set{A, \bar A}$ is observable. In the same way, a word on alphabet $\set{B, \bar B}$ corresponds to a valuation over $p_{2n}, \dots, p_{4n-1}$, thus a position in tiling B.

We also say that the inner node (non-leaf) of the binary tree is never pruned by observations (all $2n$-length words over $\set{A, \bar A, B, \bar B}$ are observable):
\begin{equation}\lbox^{< 4n} \bigwedge_{i=0..2n-1} [\Sigma]^{i} (\ldiaarg A \top \land \ldiaarg {\bar A} \top \land \ldiaarg B \top \land \ldiaarg {\bar B} \top)\end{equation}
The formula for ensuring the uniqueness of $q_t^A$ whatever the position in tiling B, and the other way around are then:
\begin{equation}
[(A+\bar A)^{2n}] \bigvee_t \lbox^{4n} \tilingprop t A \land [(B+\bar B)^{2n}] \bigvee_t \lbox^{4n} \tilingprop t B
\end{equation}

The intuition works as follows.
When evaluating $[(A+\bar A)^{2n}] \lbox^{4n} \tilingprop t A$, we consider all words $w$ in $\mathcal L((A+\bar A)^{2n})$ and we consider any pruning $\M|_w$ of the model $\modelM$ which contains the binary tree $\binarytree$. In $\M|_w$, only the leaves where the valuation on $p_0, \dots, p_{2n-1}$ that corresponds to $w$ stays. With $\bigvee_t$, we choose a tile type $t$ in $T$. The modality  $\lbox^{4n}$ then reaches all the leaves and imposes that $\tilingprop t A$ holds.

The reduction consists of computing from an instance $(\tileset, \tileseed, n)$ of the tiling problem the $\POLS$ formula $tr(\tileset, \tileseed, n)$ which is the conjunction of (1-12), which is computable in poly-time in the size of 
$(\tileset, \tileseed, n)$ (recall $n$ is in unary).
Furthermore, one can check that $(\tileset, \tileseed, n)$ is a positive instance of the tiling problem iff $tr(\tileset, \tileseed, n)$ is satisfiable.
\end{proof}

\section{Complexity results of Fragments of \texorpdfstring{$\POL^-$}{POLsatfrag}}\label{section:fragmentcomplexity}
In this section, we consider a few fragments of $\POL^-$ and we give complexity results for them. First, we consider the single agent fragment of $\POL^-$, and then we prove complexity results for the word fragment of $\POL^-$ (both single and multi-agent) using reductions to $\PAL$.
\subsection{Single agent fragment of \texorpdfstring{$\POL^-$}{POLsatfragsingleagent}}\label{subsection:starfreeonepspaceh}
\newcommand{\QBFbooleanpart}{\xi}

While we have shown (in Theorem~\ref{thm:NEXPTIME}) that the satisfiability problem of the $\POLS$ is $\NEXPTIME$-Hard, the hardness proof holds only for the case when the number of agents is at least 2. 
However, we prove that  satisfiability problem in the single Agent fragment of $\POLS$  is $\PSPACE$-Hard, although single-agent epistemic logic $S5$ is $\NP$-Complete. 

We prove it by reducing TQBF into our problem. The TQBF problem is: given a formula $\varphi$ of the form $Q_1x_1Q_2x_2\ldots Q_nx_n \QBFbooleanpart(x_1,x_2,\ldots,x_n)$ where $Q_i\in\{\forall, \exists\}$ and $\QBFbooleanpart(x_1,x_2,\ldots,x_n)$ is a Boolean formula in CNF over variables $x_1,\ldots, x_n$, decide whether the formula $\varphi$ is true.


\begin{theorem}\label{thm:singlePOL}
     The satisfiability problem for single agent fragment of $\POLS$ is $\PSPACE$-Hard.
 \end{theorem}

The proof follows in the same lines as the proof of $\PSPACE$-Hardness of the model-checking problem of the $\POLS$ (\cite{DBLP:conf/ijcai/0001GGS22}).  
We present the complete proof of Theorem~\ref{thm:singlePOL}  in the appendix.

\subsection{Word fragment of \texorpdfstring{$\POL^-$}{POLsatfragword}}\label{subsection:wordcomplexity}
To investigate the complexity of the satisfaction problem of the word fragment of $\POL^-$, we use a translation of $\POL^-$ to $\PAL$. Before going forward, let us give a very brief overview of the syntax and semantics of $\PAL$.


\subsubsection{Public announcement logic \texorpdfstring{$(\PAL)$}{PALintro}}

To reason about announcements of agents and their effects on agent knowledge, $\PAL$ \cite{DBLP:journals/synthese/Plaza07} was proposed. The underlying model that is dealt with in $\PAL$ is epistemic, $\langle S, \sim ,V\rangle$ where $S$ is a non-empty set of states, $\sim$ assigns to each agent in $\Ag$ an equivalence relation $\sim_i \subseteq S \times S$, and $V : S \rightarrow 2^{\BP}$ is a valuation function. The language is given as follows:

\begin{defi}[$\PAL$ syntax]
    Given a countable set of {\em propositional variables} $\BP$, and a finite set of {\em agents} $\Ag$, a formula $\varphi$ in Public Announcement Logic ($\PAL$) can be defined recursively as: 
    \begin{align*}
    \varphi := \top\ \  |\ \  p\ \  |\ \ \neg\phi \ \ |\ \ \phi \wedge \phi \ \ |\ \ K_i\phi\ \ |\ \ [\phi!]\phi
\end{align*}
\noindent where $p\in\BP$, and $i\in\Ag$. 
\end{defi}
Typically, $[\phi!]\psi$ says that `if $\phi$ is true, then $\psi$ holds after having publicly announced $\phi$'.
Similarly, as in $\POL^-$ syntax, the respective dual formulas are defined as,
\begin{align*}
    \hat{K_i}\psi &= \neg K_i\neg\psi\\
    \ldiaarg{\phi!}\psi &= \neg [\phi!]\neg\psi
\end{align*}

Formula $\ldiaarg{\phi!}\psi$ says that $\phi$ is true, and $\psi$ holds after announcing $\phi$.
Before going into the truth definitions of the formulas in $\PAL$, let us first define the notion of model update. 
\begin{defi}[Model Update by Announcement]
   Given an epistemic model, $\M = \langle S, \sim ,V\rangle$, $s \in S$, and a $\PAL$ formula $\phi$, the model $\M|_\phi = \langle S', \sim' ,V'\rangle$  is defined as:
    \begin{itemize}
        \item $S' = \{s\in S\mid \M,s\vDash\phi\}$
        \item ${\sim'_i}={\sim_i}|_{S'\times S'},$
        \item $V'(s) = V(s)$ for any $s\in S'$.
    \end{itemize}
\end{defi}

Now we are all set to give the truth definitions of the formulas in $\PAL$ with respect to pointed epistemic models:
\begin{defi}[Truth of a $\PAL$ formula]
    Given an epistemic model $\M = \langle S, \sim, V\rangle$ and an $s\in S$, a $\PAL$ formula $\varphi$ is said to hold at $s$ if the following holds:
    \begin{itemize}
        \item $\M,s\vDash p$ iff $p\in V(s)$, where $p\in\BP$.
        \item $\M,s\vDash\neg\phi$ iff $\M,s\nvDash\phi$.
        \item $\M,s\vDash \phi\wedge\psi$ iff $\M,s\vDash\phi$ and $\M,s\vDash\psi$.
        \item $\M,s\vDash K_i\phi$ iff for all $t\in S$ with $s\sim_i t$, $\M,t\vDash\phi$.
        \item $\M,s\vDash [\psi!]\phi$ iff $\M,s\vDash\psi$ implies$\M|_{\psi},s\vDash\phi$.
    \end{itemize}
\end{defi}

\subsubsection{On complexity}

To study the satisfiability problem for the word fragment of $\POLS$, we transfer the following result from $\PAL$ to $\POLS$:

\begin{theorem}\cite{DBLP:conf/atal/Lutz06}
The satisfiability problem of $\PAL$ is $\NP$-Complete for the single-agent case and $\PSPACE$-Complete for the multi-agent case.
\end{theorem}

$\PAL$ is the extension of epistemic logic with dynamic modal constructions of the form $[\phi!]\psi$ that expresses `if $\phi$ holds, then $\psi$ holds after having announced $\phi$ publicly'.
The dynamic operator $\ldiaarg\pi$ in the word fragment of $\POLS$ consists in announcing publicly a sequence of observations. W.l.o.g. as $\pi$ is a word $a_1\dots a_k$,  $\ldiaarg\pi$ can be rewritten as  $\ldiaarg{a_1}\dots \ldiaarg{a_k}$. In other words, we suppose that the $\POLS$ dynamic operators only contain a single letter. The mechanism of $\POLS$ is close to Public announcement logic ($\PAL$). Observing $a$ consists in announcing publicly that $wa$ occurred where $w$ is the observations already seen so far.

\newcommand{\propoccured}[1]{p_{#1}}
\newcommand{\paldia}[1]{\langle #1 ! \rangle}
\newcommand\emptyword{\epsilon}
We introduce fresh atomic propositions $\propoccured {wa}$ to say that letter $a$ is compatible with the current state given that the sequence $w$ was already observed.

For all words $w \in \Sigma^*$, we then define $tr_w$ that translates a $\POLS$ formula into a $\PAL$ formula given that $w$ is the already seen observations seen so far:
\begin{align*}
    tr_w(p) = & p \\
    tr_w(\lnot \phi) = & \lnot tr_w(\phi) \\
    tr_w(\phi \land \psi) = & tr_w(\phi) \land tr_w(\psi) \\
    tr_w(K_i \phi) = & K_i tr_w(\phi) \\
    tr_w(\ldiaarg a \phi) = & \paldia{\propoccured {wa}} tr_{wa}(\phi) 
\end{align*}
We finally transform any $\POLS$ formula $\phi$ into $tr(\phi) := tr_{\emptyword}(\phi)$.

\begin{example}
Consider the $\POLS$ formula $\phi := 
[a]\bot \land \ldiaarg{a}\ldiaarg{a}\top$. $tr(\phi)$ is $[p_a!]\bot \land \paldia{\propoccured {a}} \paldia{\propoccured {aa}} \top$. Note that if $\propoccured a$ is false, the truth value of $\propoccured{aa}$ is irrelevant.
\end{example}

\begin{proposition}\label{prop:PAL}
$\phi$ is satisfiable in the word fragment of $\POLS$ iff $tr(\phi)$ is satisfiable in $\PAL$.
\label{proposition-POLPAL}
\end{proposition}

\begin{proof}(sketch)
\fbox{$\Rightarrow$} Suppose there is a pointed $\POLS$ model $\modelM, s_0$ such that $\modelM, s_0 \models \phi$. We define $\modelM'$ to be like $\modelM$ except that for all states $s$ in $\modelM$, for all $w \in \Sigma^*$, we say that $\propoccured w$ is true at $\modelM', s$ iff $ \Exp(s)\regdiv w \neq \emptyset$. It remains to prove that $\modelM', s_0 \models tr(\phi)$. 
We prove by induction on $\phi$ that for all $w \in words(\phi)$, if $\Exp(s) \regdiv w \neq \emptyset$ then $\M|_w, s \models \phi$ iff $\M', s \models tr_w(\phi)$. 

We only show the interesting case of $\varphi = \ldiaarg{a}\psi$. Here the $\tr_w(\ldiaarg{a}\psi) = \paldia{\propoccured{wa}}\tr_{wa}(\psi)$. By assumption, $\M|_w,s\vDash\ldiaarg{a}\psi$. Hence $\M|_{wa},s\vDash\psi$. Therefore $Exp(s)\regdiv wa\neq\emptyset$. By definition of $\M'$, $p_{wa}$ is true in $s$. Therefore by IH $\M',s\vDash\tr_{wa}(\psi)$. And since $p_{wa}$ is true, hence $\M',s\vDash\ldiaarg{p_{wa}!}\tr_{wa}\psi$. Conversely, assuming $\M',s\vDash\paldia{\propoccured{wa}}tr_{wa}(\psi)$. Hence $\propoccured{wa}$ is true in $s$. By definition, $\propoccured{wa}$ is true iff $\Exp(s)\regdiv wa\neq \emptyset$. Also by IH, $\M|_{wa},s\vDash\psi$. Hence $\M|_w,s\vDash\ldiaarg{a}\psi$.

\fbox{$\Leftarrow$} Suppose there is a pointed epistemic model $\modelM', s_0$ such that $\modelM', s_0 \models tr(\phi)$. We define a $\POLS$ model $\modelM$ like $\modelM'$ except that for all states $s$, 
$\Exp(s) = \set{w \in \Sigma^* \suchthat \modelM, s \models \propoccured w}$. 
It remains to prove that $\modelM, s_0 \models \phi$. For the rest of the proof, we prove by induction on $\phi$ that for all $w \in \Sigma^*$, if $\Exp(s) \regdiv w \neq \emptyset$ then $\M|_w, s \models \phi$ iff $\M', s \models tr_w(\phi)$. The proof goes similarly as earlier.
\end{proof}

Note that the single-agent and multi-agent word fragment of $\POLS$ is a syntactic extension of propositional logic and the multi-agent epistemic logic respectively, which are $\NP$-Hard and $\PSPACE$-Hard respectively. 
From the fact that the satisfiability problem of single agent and the multi-agent fragments of $\PAL$ is in $\NP$ and $\PSPACE$ respectively, we have the following corollaries of Proposition~\ref{prop:PAL}. 

\begin{corollary}\label{coroll:singlewordNPC}
The satisfiability problem of the single-agent word fragment of $\POLS$ is $\NP$-Complete.  
\end{corollary}

    

\begin{corollary}\label{coroll:multiwordPSPACEC}
The satisfiability problem of the multi-agent Word fragment of $\POLS$ is $\PSPACE$-Complete.  
\end{corollary}

    





\section{Related work}\label{section:relatedwork}
\newcommand{\PR}{\mathsf{PR}}
\newcommand{\NM}{\mathsf{NM}}
The complexity of Dynamic Epistemic Logic with action models and non-deterministic choice of actions is $\NEXPTIME$-Complete too \cite{DBLP:conf/tark/AucherS13} and their proof is similar to the one of Theorem~\ref{theorem-nexptimehard}.

The tableau method described for $\POL^-$ uses a general technique where terms contain the observations/announcements/actions played so far. This technique was already used for PAL \cite{DBLP:journals/logcom/BalbianiDHL10}, DEL \cite{DBLP:conf/tark/AucherS13}, and for a non-normal variant of PAL \cite{DBLP:conf/icla/MaSSV15}.

Decidability of (single-agent) epistemic propositional dynamic logic ($\EPDL$) with Perfect Recall ($\PR$) and No Miracles ($\NM$) is addressed in 
\cite{DBLP:journals/logcom/Li18}. Although $\PR$ and $\NM$ are validities in $\POL^-$, there are  differences to consider even in single agent. Firstly, in an $\EPDL$ model, a possible state can execute a program $a$ and can non-deterministically transition to a state among multiple states, whereas in $\POL^-$, if a state survives after observation $a$, it gives rise to the same state except the $\Exp$ function gets residued. Also, in $\EPDL$, after execution of a program, the state changes hence the propositional valuation in the state changes, whereas in $\POL^-$, the state \emph{survives} after a certain observation and hence the propositional valuation remains the same. 

Whereas in $\POLS$, observations update the model, there are other lines of work in which specifying what agents observe define the epistemic relations in the underlying Kripke model \cite{DBLP:conf/kr/CharrierHLMS16} (typically, two states are equivalent for some agent $i$ if agent $i$ observes the same facts in the two states).  



\section{Perspectives}\label{section:perspectives}
This work paves the way to an interesting technical open question in modal logic: the connection between $\POLS$ and product modal logics. Single-agent $\POLS$ is close to the product modal logic $S5 \times K$, the logic where models are Cartesian products of an S5-model and a K-model. Indeed, the first component corresponds to the epistemic modality $\hat K_i$ while the second component corresponds to observation modalities $\ldiaarg{\pi}$. There are however two important differences. First, in $\POLS$, valuations do not change when observations are made. Second, the modality $\ldiaarg{\pi}$ is of branching at most exponential in $\pi$ while modalities in K-models do not have branching limitations. We conjecture that the two limitations can be circumvented but it requires some care when applying the finite model property of product modal logic $S5 \times K$. If this connection works, it would be a way to prove $\NEXPTIME$-Completeness of star-free single-agent $\POLS$.

Recall that $\POLS$ is close to $\PAL$ with propositional announcements only (see Proposition~\ref{proposition-POLPAL}).
We conjecture some connections between $\POLS$ and arbitrary $\PAL$ \cite{DBLP:conf/aiml/FrenchD08}, and more precisely with Boolean arbitrary public announcement logic \cite{DBLP:journals/lmcs/DitmarschF22}. Indeed, the non-deterministic choice $+$ enables to check the existence of some observation to make (for instance, $\ldiaarg{(a+b)^{10}}\phi$ checks for the existence of a 10-length word to observe), which is similar to checking the existence of some Boolean announcement.

The next perspective is also to tackle $\POL$ with Kleene-star in the language. This study may rely on techniques used in epistemic temporal logics. PAL with Kleene-star is undecidable \cite{DBLP:journals/sLogica/MillerM05}. Again, the undecidability proof relies on modal announcements. Since $\POL$ is close to Boolean announcements, this is a hope for $\POL$ to be decidable. The idea would be to exploit the link between dynamic epistemic logics and temporal logics \cite{DBLP:journals/igpl/DitmarschHR13}, and rely on techniques developed for tackling the satisfiability problem in epistemic temporal logics \cite{DBLP:journals/jcss/HalpernV89}. 


\newpage
\bibliographystyle{unsrt}
\bibliography{kr-sample}

\newpage
\appendix
\section{Proof of Soundness of the Tableau Method of \texorpdfstring{$\POLS$}{POLSTableauSound}}
In this section we give the proof the Theorem \ref{thm:starfreemultisound}.

We prove that if $\varphi$ is satisfiable then there exists a subtree rooted at some child $\Gamma_c$ of the root $\Gamma$ in $\mathcal{T}$ which is open using induction on the depth of the tableau tree $\mathcal{T}$. Let the pointed epistemic model that satisfy $\varphi$ be $\M,s$.

\textbf{Base Case.} Let the base case be $|\Gamma| = 2 + |Agt|$. Since we start from $\{(\sigma\ \ \epsilon\ \ \varphi), (\sigma, \sigma)_{i\in Agt}, (\sigma\ \ \epsilon\ \ \checkmark)\}$. Hence this implies $\varphi = l$ or $\varphi = [a]\psi$, where $l$ is a literal (a positive propositional letter or a negation of it). Hence the tableau remains open since no $\bot\in\Gamma$.

\textbf{Induction Hypothesis.} Let the statement be true for any tableau tree $\mathcal{T}$ of depth at most $n$.

\textbf{Inductive Step.} Consider the tableau tree $\mathcal{T}$ rooted at $\Gamma = \{(\sigma\ \ \epsilon\ \ \varphi), (\sigma, \sigma), (\sigma\ \ \epsilon\ \ \checkmark)\}$ of depth at most $n + 1$. Now we go case by case with $\varphi$:
\begin{itemize}
    \item $\varphi = \psi\wedge\chi$: We apply AND rule, and hence $\Gamma' = \Gamma\union\{(\sigma\ \ \epsilon\ \ \psi), (\sigma\ \ \epsilon\ \ \chi)\}$, which is a child rooted at $\Gamma$. Since $\M,s\vDash\varphi$, by definition $\M,s\vDash\psi$ and $\M,s\vDash\chi$. By IH, tableau tree rooted at $\Gamma'$ is open, which suggestes, $\mathcal{T}$ is open.
    
    \item $\varphi = \psi\vee\chi$: We apply OR rule and hence we get two children, $\Gamma_1 = \Gamma\union\{(\sigma\ \ \epsilon\ \ \psi)\}$ and $\Gamma_2 = \Gamma\union\{(\sigma\ \ \epsilon\ \ \chi)\}$. Since $\M,s\vDash\psi\vee\chi$, hence $\M,s\vDash\psi$ or $\M,s\vDash\chi$. By IH, one of the sub tableau tree rooted at $\Gamma_1$ or $\Gamma_2$ will be open, hence implying $\mathcal{T}$ to be open.
    
    \item $\varphi = K_i\psi$: By applying the Knowledge rule, $\Gamma' = \Gamma\union\{(\sigma'\ \ \epsilon\ \ \psi)\mid \{(\sigma, \sigma')_i, (\sigma'\ \ \epsilon\ \ \checkmark)\}\subseteq\Gamma\}$. Since $\M,s\vDash K_i\psi$, hence for every $s'\in\M$ such that $s\sim_i s'$, $\M,s'\vDash\psi$. By IH, the tableau subtree rooted at $\Gamma'$ is open.
    
    \item $\varphi = \hat{K_i}\chi$: By applying the possibility rule, $\Gamma' = \Gamma\union\{(\sigma_n, \sigma_n)_{i\in Agt}, (\sigma_n\ \ \epsilon\ \ \checkmark), (\sigma_n\ \ \epsilon\ \ \chi)\}\union\{(\sigma, \sigma_n)_i, (\sigma_n, \sigma)_i\}\union\{(\sigma', \sigma_n)_i, (\sigma_n, \sigma')_i\mid \{(\sigma'\ \ \epsilon\ \ \checkmark), (\sigma,\sigma')_i\}\subseteq\Gamma_l\}$. Since, $\M,s\vDash\hat{K_i}\chi$, hence there is an $s_n\in\M$ such that $\M,s_n\vDash\chi$, also $s_n\sim s'$ for every $s'\in\M$ since $\sim$ is an equivalence relation. Hence by IH, the sub tableau tree rooted at $\Gamma'$ is open.

    \item $\varphi = \ldiaarg{\pi\pi'}\psi$:  Hence by the diamond decomposition rule $\Gamma' = \Gamma\union\{(\sigma\ \ \epsilon\ \ \ldiaarg{\pi}\ldiaarg{\pi'}\psi)\}$. Since $\M,s\vDash\ldiaarg{\pi\pi'}\psi$, hence $\M|_{w_*},s\vDash\psi$, for some $w_*\in\LL(\pi\pi')$ which also suggests, there is a $w\in\LL(\pi)$ and a $w'\in\LL(\pi')$ such that $w_* = ww'$. Hence   $\M|_w,s\vDash\ldiaarg{\pi'}\psi$, which implies $\M,s\vDash\ldiaarg{\pi}\ldiaarg{\pi'}\psi$. By IH, the tableau for $\Gamma'$ is open.

    \item $\varphi = \ldiaarg{a}\psi$. Hence now by Projection rule, $\Gamma' = \Gamma\union\{(\sigma\ \ a\ \ \checkmark), (\sigma\ \ a\ \ \psi)\}$. Since $\M,s\vDash\ldiaarg{a}\psi$, therefore $s\in\M|_a$ and $\M|_a,s\vDash\psi$. And hence by IH, tableau tree rooted at $\Gamma'$ is open. 

    \item $\varphi = [\pi]\psi$. Say $(\sigma\ \ a\ \ \checkmark)$ is added for some diamond formula term, say of the form $(\sigma\ \ \epsilon\ \ \ldiaarg{a}\psi')$, hence the proof has added $(\sigma\ \ a\ \ [\pi\regdiv]\psi)$. By assumption $\M, s\vDash\ldiaarg{a}\psi'$, and hence $s\in\M|_a$. Therefore $\M|_a,s\vDash[\pi\regdiv a]\psi$, hence the proof starting from terms $\{(\sigma, \sigma), (\sigma\ \ a\ \ \checkmark), (\sigma\ \ a\ \ [\pi\regdiv a]\psi)\}$ will remain open.

    For the other case $(\sigma\ \ a\ \ \checkmark)$ is not added, hence $s\notin\M|_a$, therefore the tableau remains open.
\end{itemize}
The box modality case goes similarly as in case of single agent.

\section{The \texorpdfstring{$\NEXPTIME$}{POLSTableauAlgo} Algorithm for Star-Free Multi-Agents}

Now we design an algorithm based on tableau and prove existence of an algorithm that takes non-deterministically exponential steps with respect to the size of $\varphi$. Now given a $\varphi$, we now create a tree of nodes that contains terms of the form $(w, \psi)$ and $(w, \checkmark)$, where $w\in \Sigma^*$ is a word that is occuring in tableau, and $\psi$ is a formula in $FL(\varphi)$. Each node $T_\sigma$ refers to a state label $\sigma$ in tableau, a term of the $(w,\psi)\in T_\sigma$  intuitively translates to in the state corresponding to $\sigma$, after projecting model on $w$, the state survives and there $\psi$ is satisfied, and hence refers to the term $(\sigma\ \ w\ \ \psi)$ in tableau. Similarly, $(w, \checkmark)\in T_\sigma$ means state corresponding to $\sigma$ survives after projection on $w$, and hence refers to the term $(\sigma\ \ w\ \ \checkmark)$ in the tableau. The tableau tree created, we call it $\cT_\tP$

For this algorithm, we change the definition of saturation and unsaturation a bit from the earlier definition. We say $T_\sigma$ is \textbf{unsaturated} against a rule $R$ iff there is a term in $(w, \psi)\in T_\sigma$ or $(w, \checkmark)\in T_\sigma$, such that $(\sigma\ \ w\ \ \psi)$ or $(\sigma\ \ w\ \ \checkmark)$ lies in the numerator of $R$ but there is no denominator $(\sigma\ \ w\ \ \psi')$ of $R$ such that $(w, \psi')$ is in $T_\sigma$, similar for terms like $(w, \checkmark)$. We call the term $(w, \psi)$ or $(w, \checkmark)$ here to be the \textbf{reason for unsaturation}.

We saturate the rules carefully such that each node in the tree corresponds to a single state in the model. This technique is well studied in \cite{DBLP:journals/ai/HalpernM92}.



\begin{breakablealgorithm}
\caption{StarFree-SAT}
\begin{algorithmic}[1]
\Procedure{StarFree-SAT}{$\varphi$}
    \State $T_{\sigma_0} \gets  \{(\epsilon, \varphi), (\epsilon, \checkmark)\}$ 
    \State $T_{\sigma_0}$ is the root of tree $\cT_P$.
    \While{there is a leaf of $\cT_P$ that satisfies one of the following conditions}
        \If{$\bot\notin T_\sigma$ and $T_\sigma$ is unsaturated against Propositional and Survival Rules then}
            \State Let $(w, \psi)$ or $(w, \checkmark)$ be the reason for unsaturation against the above rules.
            
            \If{$\psi = \psi_1\wedge\psi_2$} $T_\sigma = T_\sigma\union\{(w, \psi_1),(w, \psi_2)\}$
            
            \ElsIf{$\psi = \psi_1\vee\psi_2$} non deterministically choose $\psi_1$ or $\psi_2$ and $T_{\sigma} = T_\sigma\union\{(w, \psi_1)\}$ or $T_{\sigma} = T_\sigma\union\{(w, \psi_2)\}$ as per choice.
            
            \ElsIf{$\psi = l$ where $l$ is a literal} $T_\sigma = T_\sigma\union\{(w, [a]l), (\epsilon, l)\}$.

            \ElsIf{$\{(w, \psi), (w, \neg\psi)\}\subseteq T_\sigma$} $T_\sigma = T_\sigma\union\{\bot\}$.

            \ElsIf{$(w, \checkmark)\in T_\sigma$} $T_\sigma = T_\sigma\union\{(w',\checkmark)\mid w'\in Pre(\LL(w))\}$.

            \EndIf

        \ElsIf{$\bot\notin T_\sigma$ and $T_\sigma$ is propositionally saturated but unsaturated against Box and Diamond Rules}
            \State Let the reason for unsaturation be $(w, \psi)$.

            \State If $\psi = \ldiaarg{\pi_1\pi_2}\psi'$ then $T_\sigma = T_\sigma\union\{(w,\ldiaarg{\pi_1}\ldiaarg{\pi_2}\psi')\}$.

            \State If $\psi = \ldiaarg{\pi_1 + \pi_2}\psi'$ then nondeterministically choose a $\pi_i$, where $i\in\{1,2\}$ and make $T_\sigma = T_\sigma\union\{(w, \ldiaarg{\pi_i}\psi')\}$

            \State If $\psi = \ldiaarg{a}\psi'$ then $T_\sigma = T_\sigma\union\{(wa, \checkmark), (wa, \psi'), (w, [a]\psi')\}$.

            \State If $\psi = [\pi]\psi'$ and $(wa,\checkmark)\in T_\sigma$ then $T_\sigma = T_\sigma\union\{(wa, [\pi\regdiv a]\psi')\}$

        \ElsIf{$\bot\notin T_\sigma$ and it saturated against Propositional, Survival and Box, Diamond Rules but for all formula $\psi$ that occurs in terms $(w, \psi)\in T_\sigma$, there is a $\psi'\in FL(\psi)$ such that neither $\psi'$ nor $\neg\psi'$ occurs in $T_\sigma$}
            \State Non-deterministically choose $\psi'$ or $\neg\psi'$ and $T_\sigma = T_\sigma\union\{(w,\psi')\}$ or $T_\sigma = T_\sigma\union\{(w,\neg\psi')\}$ as per choice, where $\neg\psi'$ is in Negation normal form.

        \ElsIf{$\bot\notin T_\sigma$ and $T_\sigma$ is saturated against Propositional rules, Survival Rules and Diamond, Box rules and there is no formula $\psi$ that occurs as $(w, \psi)\in T_\sigma$, there is a $\psi'\in FL(\psi)$ such that neither $\psi'$ nor $\neg\psi'$ occurs in $T_\sigma$}
            \ForEach{$\hat{K_i}\psi$ that occurs in $T_\sigma$} 
                \State $T = \{(w, \hat{K_i}\psi')\mid (w, \hat{K_i}\psi')\in T_\sigma\}\union \{(w, K_i\psi')\mid (w, K_i\psi')\in T_\sigma\}\union\{(w, \psi)\mid (w, \hat{K_i}\psi)\in T_\sigma\}$
                \State If there is no $i$-ancestor $T_{\sigma''}$ of $T_\sigma$ such that $T\subseteq T_{\sigma''}$ then add a $i$-child $T_{\sigma'} = T$
            \EndFor
        \EndIf
    \EndWhile
    \While{$T_{\sigma_0}$ is not marked}
        \If{there is an unmarked leaf node $T_\sigma$ of $\cT_P$}
            \If{$\bot\in T_\sigma$ or $\{(w, K_i\psi), (w, \neg\psi)\}\subseteq T_\sigma$}
                \State Mark $T_{\sigma_0}$ "UNSAT"
            \Else
                \State Mark $T_{\sigma_0}$ "SAT"
            \EndIf

        \Else
            \State $T_\sigma$ is an unmarked internal node whose all children are marked.
            \If{ all children of $T_\sigma$ is marked "SAT"}
                \State Mark $T_\sigma$ "SAT"
            \Else
                \State Mark $T_\sigma$ "UNSAT"
            \EndIf
        \EndIf
    \EndWhile
    \If{$T_{\sigma_0}$ is marked "SAT"}
        \State Return SAT
    \Else 
        \State Return UNSAT
    \EndIf
    \EndProcedure
\end{algorithmic}
\end{breakablealgorithm}

\section{Satisfiability problem of single agent fragment of \texorpdfstring{$\POLS$ is $\PSPACE$}{POLSoneagentHardness}-Hard}

In this section we prove Theorem~\ref{thm:singlePOL}
    \newcommand{\letterQBF}[1]{a_{x_{#1}}}
    \newcommand{\letterQBFneg}[1]{a_{\overline{x_{#1}}}}
    \newcommand{\QBFformula}{\varphi}
    \newcommand{\reduction}{\tau}
    \newcommand{\finalformula}{\reduction(\QBFformula)}
    \newcommand{\literal}{\ell}

We will prove the hardness by reduction from TQBF. Given a QBF formula $\varphi = \exists x_1 \forall x_2 \ldots Q_nx_n\QBFbooleanpart(x_1,\ldots, x_n)$, where $Q_i$ is $\exists$ if $i$ odd, and is $\forall$ is even, and $\QBFbooleanpart(x_1,\ldots, x_n)$ is a propositional formula in CNF over variables $x_1,\ldots,x_n$. Without loss of generality, we suppose we have $m$ clauses with at most 3 literals in each. The objective is to define a $\POLS$-formula $\reduction(\QBFformula)$, computable in poly-time in $|\QBFformula|$, such that $\QBFformula$ is QBF-true iff  $\reduction(\QBFformula)$ is $\POLS$-satisfiable. To save space, we denote $\neg x_i$ as $\overline{x_i}$, where $x_i$ is a variable in the TQBF. We also write  $\overline{\overline{x_i}}$ for $x_i$.

    \paragraph{Definition of $\reduction(\QBFformula)$.}
    We encode valuations over $x_1, \dots, x_n$ by words on the alphabet $\set{\letterQBF i, \letterQBFneg i \mid i=1..n}$. We say that a literal $\literal_h$ is consistent with a word $w$ if $a_{\literal_h}$ appears in $w$. For instance the word $\letterQBF 1 \letterQBFneg 2 \letterQBFneg 3$ encodes the valuation in which $x_1$ is true and both $x_2$ and $x_3$ are false. Set of valuations are represented by languages. For $1 \leq u \leq v \leq n$,
    $B^u_v := (\letterQBF{u+1} + \letterQBFneg{u+1})(\letterQBF{u+2} + \letterQBFneg{u+2})\ldots(\letterQBF{v} + \letterQBFneg{v})$, (by convention $B^u_v = \epsilon$ when $u=v$).
    Intuitively, $B^u_v$ is the language encoding the set of all possible valuations over propositions $x_{u+1}, \dots, x_v$.

    \begin{itemize}


        \item   We first define several formulas to express constraints on expectations:
        \begin{itemize}
            \item The formula $T_i := (\ldiaarg{\letterQBF{i}}\top\wedge\ldiaarg{\letterQBFneg{i}}\top)$
             imposes that the current state survives after observing both $\letterQBF{i}$ as well as $\letterQBFneg{i}$.

            \item For each literal $\literal_h$ being $x_h$ or $\overline{x_h}$, we build a formula~$L_h$ that enforces the expectation at the current state contains all words encoding valuations over propositions $x_1, \dots, x_n$ in which $\literal_h$ is true:
            \begin{align*}
                L_h := &\bigwedge_{i = 1}^{h-1}([B^0_{i-1}]T_i\\
                    \wedge&[B^0_{h-1}](\ldiaarg{a_{\literal_h}}\top\wedge [a_{\overline{\literal_h}}]\bot)\\
                    \wedge&\bigwedge_{i=h+1}^n[B^0_{h-1}a_{\literal_h}B^{h}_{i-1}]T_i
            \end{align*}
            \item Finally for each clause $C_j = (\literal_h\vee \literal_r\vee \literal_k)$, we define the formula  $\tr(C_j) := K(p_j\rightarrow(L_h\vee L_r\vee L_k))\wedge \hat{K}p_j$.
            The subformula $\hat{K}p_j$ enforces the existence of a $p_j$-state. And the subformula $K(p_j\rightarrow(L_h\vee L_r\vee L_k))$ enforces that any $p_j$-state survives on all the words from $w\in B^0_n$ which are consistent with either $l_h$, or $l_r$ or $l_k$.
            
        \end{itemize}

        \item $S := \tr(Q_1x_1)\tr(Q_2x_2)\ldots\tr(Q_nx_n)\bigwedge_{j=1}^m\hat{K}p_j$ where $\tr(\forall x_i) = [\letterQBF{i} + \letterQBFneg{i}]$, $\tr(\exists x_i) = \ldiaarg{\letterQBF{i} + \letterQBFneg{i}}$ for any $i\in[n]$.
        Intuitively, the choice of a valuation over $x_1, \dots, x_n$ by the two players $\exists$ and $\forall$ in the QBF-prefix $Q_1x_1\dots Q_n x_n$ 
        is simulated by the choice of a word in $B^0_n$ by the two players $\ldiaarg{.}$ and $[.]$
        so that all clauses are true (i.e. all $p_j$-state survives the observation of $w$: $\bigwedge_{j=1}^m\hat{K}p_j$).
    
    \end{itemize}
    
The $\POLS$-formula $\finalformula$ is defined by 
\begin{align*}
            \finalformula = \bigwedge_{C_j\in\varphi}\tr(C_j)\wedge S\wedge (\bigvee_{j=1}^m p_j)
        \end{align*} 
    Note that $\finalformula$ can be computed in poly-time in $|\varphi|$.
    Let us prove that $\varphi$ is a QBF-true iff $\finalformula$ is $\POLS$-satisfiable.

\fbox{$\Rightarrow$}    First assume $\varphi$ is QBF-true. Hence there is a Quantifier tree that is certifying the truth. We create a model:
    \begin{itemize}
        \item $W = \{1,2,\ldots, m\}$
        \item $R = W\times W$ 
        \item $V(j) = \{p_j\}$
        \item $Exp(j) = \sum_{l_i\in C_j}(\Pi_{k=1}^{i-1}(a_k + \bar{a_k})\tr_c(l_i)\Pi_{k=i+1}^n(a_k + \bar{a_k}))$
    \end{itemize}

First we prove $\M_\varphi,1\vDash\tr(C_j)$ for every $j\in[m]$.

Consider for any state $j$, since $V(j) = \{p_j\}$, hence $\hat{K}p_j$ stands true from any state. Now we prove, since the formula $K(p_j\rightarrow (\tr_m(l_h)\vee \tr_m(l_i)\vee \tr_m(l_k)))$ insists, that $\M_\varphi,j\vDash(\tr_m(l_h)\vee \tr_m(l_i)\vee \tr_m(l_k))$

Given $C_j = (l_h\vee l_k\vee l_r)$, since $\varphi$ is true, there is a path (among many) is the quantifier tree where at the end of the path (leaf node), it is true that in every clause at least one literal is assigned true. Let us consider in $C_j$, in this path, $l_h$ was assigned true. 

By induction on $1\leq i< h$, it can be proved that $\M_\varphi|_{w},j\vDash T_i$ for every $w\in\LL(B^1_{i-1})$.

By the definition of the model and by the fact that for every $w\in\LL(B^1_{h-2})$, $\M_\varphi|_{w},j\vDash T_{h-1}$, it can be seen that $\M_\varphi|_{w},j\vDash([a_{\overline{l_h}}]\bot\wedge\ldiaarg{a_{l_h}}\top)$ for every $w\in\LL(B^1_{h-1})$


 Again, by induction on $h < i\leq n$, it can be proved that, $\M_\varphi|_w,j\vDash T_i$ for every $w\in\LL(B^0_{h-1}a_{l_h}B^h_{i-1})$.

Hence $\M_\varphi,1\vDash\tr(C_j)$.

Now we prove If $\varphi$ is true then $\M_\varphi,1\vDash\tr_c(Q_1x_1)\tr_c(Q_2x_2)\ldots\tr_c(Q_nx_n)\bigwedge_{j=1}^m \hat{K}p_j$. 

We prove for any $i\in\{0,\ldots, n-1\}$, $\M_\varphi|_{\tr(l_1)\ldots\tr(l_{n-i})},1\vDash\tr(Q_{n-i+1}x_{n-i+1})\ldots\tr(Q_nx_n)\bigwedge_{j=1}^m\hat{K}p_j$, where $(l_1,\ldots,l_{n-i})$ represent any assignment respect to the true paths in the quantified boolean tree upto level $n-i$.

\textbf{Base Case} Consider the case for $i=1$. Since by assumption $n$ is even, $\tr(Q_nx_n) = [\letterQBF{n} + \letterQBFneg{n}]\bigwedge_{j=1}^m\hat{K}p_j$. Since by assumption, $(l_1,\ldots,x_n)$ and $(l_1,\ldots,\bar{x_n})$ are a satisfying assignment for $\xi$, hence $\M_\varphi|_{\tr(l_1)\ldots\letterQBF{n}}, 1\vDash\bigwedge_{j=1}^m\hat{K}p_j$ as well as $\M_\varphi|_{\tr(l_1)\ldots\letterQBFneg{n}},1\vDash\bigwedge_{j=1}^m\hat{K}p_j$.


\textbf{Inductive Step} Consider $i = k+1$. Consider the case where $i$ is even. Hence $\tr(Q_{n-i+1}x_{n-i+1}) = \ldiaarg{\letterQBF{n-i+1} + \letterQBFneg{n-i+1}}$. Therefor by assumption $(l_1,\ldots, x_{n-i+1})$ or $(l_1,\ldots, \bar{x_{n-i+1}})$ is an assignment that is making $\xi$ true. Hence by IH , $\M_\varphi|_{\tr(l_1)\ldots\tr(l_{n-i})\letterQBF{n-i+1}}, 1\vDash\tr(Q_{n-i+2}x_{n-i+2})\ldots\tr(Q_nx_n)\bigwedge_{j=1}^m\hat{K}p_j$ or $\M_\varphi|_{\tr(l_1)\ldots\tr(l_{n-i})\letterQBFneg{n-i+1}},1\vDash\tr(Q_{n-i+2}x_{n-i+2})\ldots\tr(Q_nx_n)\bigwedge_{j=1}^m\hat{K}p_j$, which implies $\M_\varphi|_{\tr(l_1)\ldots\tr(l_{n-i})}, 1\vDash\ldiaarg{\letterQBF{n-i+1} + \letterQBFneg{n-i+1}}\tr(Q_{n-i+2}x_{n-i+2})\ldots\tr(Q_nx_n)\bigwedge_{j=1}^m\hat{K}p_j$.


\fbox{$\Leftarrow$}
Now assume $\finalformula$ has a model $\M$ such that $\M,s\vDash\finalformula$. Now we derive the quantifier tree certifying $\varphi$ to be true.

We prove the following:
\begin{proposition}
If $t\in\M|_{\tr(l_1)\ldots\tr(l_n)}$ and $M,t\vDash p_j$ then $(l_1,\ldots,l_n)$ is  a satisfying assignment for clause $C_j$.
\end{proposition}
\begin{proof}
    Without loss of generality, let us consider $C_j = (l'_h\vee l'_r\vee l'_k)$. Hence to $\sigma = (l_1,\ldots,l_n)$ to be a satisfying assignment for $C_j$, at least one of $l'_h, l'_r$ or $l'_k$ should be in $\sigma$. 

    Suppose $(l_1,\ldots,l_n)$ is not a satisfying assignment. Hence $l_h = \bar{l'_h}, l_r = \bar{l'_r}$ and $l_k = \bar{l'_k}$. Also by assumption,  $p_j$ is true in $t$. Therefore either $L'_h$ or $L'_r$ or $L'_k$ is true here.  Consider the term $L'_h$. By definition the term $[B^0_{h-1}](\ldiaarg{a_{l'_h}}\top\wedge[a_{\overline{l'_h}}]\bot)$ is ANDed and hence is true, but this cannot be true since after projecting on $\tr(l_1)\ldots\tr(l_{h-1})$, $B^0_{h-1}\regdiv \tr(l_1)\ldots\tr(l_{h-1})$ is non-empty and hence $\M|_{\tr(l_1)\ldots\tr(l_{h-1})},t\vDash(\ldiaarg{a_{l'_h}}\top\wedge[a_{\overline{l'_h}}]\bot)$. But this is a contradiction since $\tr(l_h) = \tr(\bar{l'_h}) = a_{\overline{l'_h}}$. 
\end{proof}

\begin{proposition}
    For any $1\leq i\leq n$,
    If $s\in\M|_{\tr(l_1)\ldots\tr(l_{n-i})}$ and $\M|_{\tr(l_1)\ldots\tr(l_{n-i})}, s\vDash \tr(Q_{n-i+1}x_{n-i+1})\ldots\tr(Q_nx_n)\bigwedge_{j=1}^m\hat{K}p_j$ then $Q_{n-i+1}x_{n-i+1}\ldots Q_nx_n\xi|_{(l_1,\ldots,l_{n-i})}$ is true.
\end{proposition}
\begin{proof}
    We state that the statement as the Induction Hypothesis. Now we prove the Base Case for it, that is $i = 1$.

    \textbf{Base Case.} By assumption $s$ survives in the projection and $\M|_{\tr(l_1)\ldots\tr(l_{n-1})},s\vDash[\letterQBFneg{n} + \letterQBF{n}]\bigwedge_{j=1}^m\hat{K}p_j$. By proposition 1, since at least one state with $p_j$ for each $j$ is surviving, hence all the clause is still surviving in $\xi|_{(l_1,\ldots,l_{n-1})}$. Since $s$ has at least one of $p_j$ true here and because of the $K$ formula, $s$ survives on both projection on $\letterQBF{n}$ as well as $\letterQBFneg{n}$, and hence after that at least one state where $p_j$ is true surviving for each $j$. Hence $\forall x_n\xi_{l_1,\ldots,l_{n-1}}$ is true.

    \textbf{Inductive Step.} The Inductive Step is similarly proven as in base case.
\end{proof}

\end{document}